\documentclass[a4paper, 10pt, final, BCOR=5mm, pdftex, pagesize]{amsart}%
\usepackage[a4paper,left=3cm,right=3cm,top=2.9cm,bottom=3cm]{geometry}
\usepackage[utf8]{inputenc} %
\usepackage{fancyhdr}
\usepackage{csquotes} %
\usepackage{verbatim}
\usepackage{underscore}
\usepackage{enumitem}
\usepackage{color}

\usepackage{lmodern} %
\usepackage[T1]{fontenc} %
\usepackage{setspace}
\usepackage{pdflscape}
\usepackage[colorinlistoftodos]{todonotes}

\usepackage{amsmath}
\usepackage{amsfonts}
\usepackage{amssymb}
\usepackage{mathrsfs}
\usepackage{mathtools}
\usepackage{dsfont} %
\usepackage{bigints} %
\usepackage{amsthm}
\usepackage{accents}

\usepackage{float} %
\usepackage[margin=0.8cm]{caption} %
\usepackage{subcaption}
\usepackage{graphicx}
\usepackage{pgfplots}
\pgfplotsset{compat=1.8}
\usepackage[draft=false,babel,tracking=true,kerning=true,spacing=true]{microtype} 
\usepackage{tabularx} 

\usepackage{booktabs}

\definecolor{CMblue}{RGB}{0, 153, 153}
\definecolor{CMgreen}{RGB}{76, 153, 0}
\definecolor{CMorange}{RGB}{237, 125, 49}
\definecolor{CMlightblue}{RGB}{204, 255, 255}

\usepackage{chngcntr}
\usepackage[section]{placeins} %
\numberwithin{equation}{section}

\usepackage[noadjust]{cite}
\usepackage{hyperref}

\usepackage{tikz}
\usetikzlibrary{shapes.geometric, patterns}
\pgfplotsset{compat=1.8}
\usepackage{setspace} %
\usepackage[draft=false,babel,tracking=true,kerning=true,spacing=true]{microtype} %

\setitemize{leftmargin=0.5cm}

\newcommand{\R}{\mathbb{R}}
\newcommand{\C}{\mathbb{C}}
\newcommand{\N}{\mathbb{N}}

\newcommand{\E}{\ensuremath\mathbb{E}}

\renewcommand{\leq}{\leqslant}
\renewcommand{\geq}{\geqslant}
\renewcommand{\bar}{\overline}
\renewcommand{\tilde}{\widetilde}

\newtheorem{The}{Theorem}[section]
\newtheorem*{The*}{Theorem}

\newtheorem{Lem}[The]{Lemma}
\newtheorem{Cor}[The]{Corollary}
\newtheorem{Rmk}[The]{Remark}

\catcode`,\active

\catcode`\,12

\providecommand{\keywords}[1]
{
  \small	
  \textbf{\textit{Keywords---}} #1
}

\providecommand{\subjclass}[2]{%
 \small	
  \textbf{\textit{2010 Mathematics subject classification---}} #1
}

\newcommand{\dd}{d}

\title{Expanding the rough Heston model in $H$}

\author{Paul P. Hager}
\address{Paul P. Hager, Department of Statistics and Operations Research, University of Vienna}
\email{paul.peter.hager@univie.ac.at}

\author{D\"orte Kreher}
\address{D\"orte Kreher, Institut für Mathematik, Humboldt-Universit\"at zu Berlin}
\email{doerte.kreher@hu-berlin.de}

\thanks{The authors thank the 
participants of the Rough Volatility Day in Nov 2025 for their feedback and 
insightful comments. Especially, they are grateful to Jim Gatheral for 
encouraging them to also look at expansions around $H_0=0$. D\"orte Kreher acknowledges funding by the Deutsche Forschungsgemeinschaft (DFG, German Research Foundation) – CRC/TRR 388 "Rough Analysis, Stochastic Dynamics and Related Fields“ – Project ID 516748464 (Project B02).}



\begin{document}

\date{\today}

\begin{abstract}\small
We study the dependence of the fractional Riccati equation in the rough Heston model on the Hurst parameter $H$. For each expansion point $H_0\in(-1/2,1/2]$, we derive a Taylor expansion of the Riccati solution in $H$, whose coefficients are characterized recursively as solutions of linear Volterra equations with fractional-logarithmic kernels. We prove local uniform convergence of the resulting Taylor series and, in particular, analyticity of the fractional Riccati solution in the Hurst parameter. Through the affine transform formula, this yields approximations of the rough Heston characteristic function and Fourier prices. Numerically, once a reference solution at $H_0$ is available, the expansion coefficients can be computed recursively and evaluated for many nearby values of $H$. We implement the method around $H_0=1/2$, using the classical Heston solution, and around $H_0=0$, using a Pad\'e approximation. Experiments for European call options indicate that low expansion orders already provide accurate implied volatilities across a wide range of Hurst parameters, including the hyper-rough regime.
\end{abstract}

\maketitle

\raggedbottom

\section{Introduction}\label{sec:intro}%

Rough volatility models, originating in the empirical discovery in \cite{GatheralJaissonRosenbaum2018} and framed for pricing in early works such as \cite{BayerFrizGatheral2016}, are distinguished from other volatility models by the lower-than-Brownian H\"older regularity of the volatility process and by their ability to reproduce the steep implied volatility skew observed at short maturities.
Over the last decade, rough volatility has gathered substantial research, see the book \cite{bayer2023rough} and the references therein, as well as debate \cite{FTW19,ContDas24,Fukasawa21}, and has become one of the established modeling frameworks for equity volatility.

The most workable instance of this class of models is the \emph{rough Heston model}. First introduced and motivated through a microstructural scaling limit in \cite{EER18,EER19}, this model is particularly popular because, like its classical Heston predecessor \cite{H93}, it retains an affine structure: the characteristic function of the log-price can be obtained as an exponentially affine functional of the solution to a fractional Riccati equation.

Specifically, the risk-neutral dynamics of a rough Heston model with Hurst parameter $H\in(0,1/2)$ are built from a two-dimensional standard Brownian motion $(B,W)$ and the fractional kernel
\[
K^{H}(t)=\frac{t^{H-1/2}}{\Gamma(H+1/2)},\quad t>0,
\]
via the stochastic Volterra-integral equations
\begin{align*}
    dS_t^{H}&=\sqrt{V^{H}_t}S_t^{H}\left(\rho dW_t+\sqrt{1-\rho^2}dB_t\right),\\
    V_t^{H}&=V_0+\int_0^tK^{H}(t-s)(\theta-\lambda V^{H}_s)ds+\int_0^tK^{H}(t-s)\nu\sqrt{V_s^{H}}dW_s,
\end{align*}
where $\theta>0$, $\lambda\geq\nu>0$, $\rho\in(-1,1)$, and $V_0\geq0$, $S_0>0$. 
The extended characteristic function of the log-price $X_T:=\log(S^{H}_T/S^{H}_0)$ is given by
\begin{equation}\label{eq:charfct}
\varphi^H(z):=\E\left[e^{zX_T}\right]=\exp\left(\int_0^TF\left(z,\psi^{H}(T-t,z)\right)g^{H}(t)dt\right)
\end{equation}
for all $z\in \C^*:=\{z=a+ib:a\in[0,1],b\in\R\}$, where
\begin{align}
    \label{def:gH}g^{H}(t)&:=V_0+\theta\int_0^tK^{H}(s)ds=V_0+\theta\frac{t^{H+1/2}}{\Gamma(H+3/2)},\\
    \label{def:F}F(z,x)&:=\frac{1}{2}(z^2-z)+(\rho\nu z-\lambda)x+\frac{\nu^2}{2}x^2,
\end{align}
and $\psi^{H}(\cdot,z)$ is the unique continuous solution to the fractional Riccati equation
\begin{equation}\label{eq:riccati}
    \psi^{H}(t,z)=\int_0^tK^{H}(t-s)F\left(z,\psi^{H}(s,z)\right)ds,\quad t\in[0,T].
\end{equation}
The Riccati equation \eqref{eq:riccati} is well-posed for all $H>-1/2$, cf.~\cite{EER18,EER19}; the range \(H\in(-1/2,0]\) corresponds to the hyper-rough regime of
\cite{AJ2021}, where the model is formulated through the integrated variance
process, since the kernel singularity is integrable but not square-integrable.

The representation \eqref{eq:charfct} reduces Fourier pricing in the rough Heston model to the computation of $\psi^{H}$. This is precisely where the numerical difficulty enters. For $H=1/2$, equation \eqref{eq:riccati} is the classical Heston Riccati equation and has a closed form solution. For $H\neq1/2$, however, no closed form solution is available, and one has to solve a fractional Riccati equation with a weakly singular Volterra kernel. This equation is nonlinear, nonlocal in time and has to be evaluated at a wide range of Fourier arguments, so that both, stability and computational cost become relevant for option pricing and calibration. Existing approaches include fractional predictor-corrector and Adams-type schemes \cite{diethelm_ford_freed_2002,jeng_kilicman_2020}, fast hybrid and product-integration schemes \cite{callegaro_grasselli_pages_2020}, global collocation and quasi-linearization methods \cite{DastgerdiBastani20}, Markovian and multifactor approximations of the Volterra kernel \cite{AJEE19,BB23}, and rational or Pad\'e-type approximations of the fractional Riccati solution \cite{gatheral_radoicic_2019,gatheral_radoicic_2024}.

The purpose of this paper is complementary to these approaches. Instead of solving \eqref{eq:riccati} from scratch for each value of $H$, we study the dependence of the fractional Riccati solution on the Hurst parameter itself. Writing
\[
    \psi:[0,T]\times (-1/2,\infty)\times \C^*\rightarrow\C,
    \qquad
    (t,H,z)\mapsto \psi(t;H,z),
\]
for the unique solution of \eqref{eq:riccati}, the main idea is to approximate $\psi(t;H,z)$ by its Taylor polynomial in $H$ around a fixed expansion point $H_0$, namely
\begin{equation*}
P_{N,H_0}(t;H,z):=\sum_{n=0}^N\partial_H^n\psi\left(t; H_0,z\right)\frac{(H-H_0)^n}{n!}.
\end{equation*}
We first verify that $\psi$ is infinitely often continuously differentiable in $H$ and identify the coefficients $\partial_H^n\psi(t;H_0,z)$ recursively as the solutions of \emph{linear} Volterra equations. We then prove that the Taylor series converges locally uniformly around every expansion point $H_0\in(-1/2,1/2]$ and, in particular, that $H\mapsto\psi(t;H,z)$ is analytic.

Apart from being of independent interest for sensitivity analysis with respect to $H$, this result directly gives rise to a numerical scheme. Indeed, the linear Volterra equations solved by the expansion terms involve tractable fractional-logarithmic kernels, and their forcing terms are constructed recursively from the previously computed expansion terms. Thus, starting from a reference solution $\psi^{H_0}$, the expansion coefficients can be computed efficiently. 
Using an explicit product-integration scheme for these linear Volterra equations, the expansion terms are obtained through matrix multiplications and can then be evaluated efficiently against several values of $H$ at once.

This is particularly useful for the expansion point $H_0=1/2$, where the classical Heston solution is known explicitly and the linear Volterra equations for the expansion coefficients reduce to linear ODEs. Our numerical experiments show, perhaps surprisingly, that the expansion around $H_0=1/2$ remains stable even for values of $H$ far below the expansion point and can be used for option pricing, for instance with $H=0.2$, at maturities $T=0.083$ or larger.

Further experiments show that the expansion around $H_0=0$ is especially effective in the hyper-rough regime. In this case, the reference solution can be obtained efficiently from a Pad\'e approximation, and the resulting expansion remains stable for negative values of $H$ down to $H=-0.3$ at maturities $T=0.019$ or larger, a regime which is reported to be particularly challenging for direct discretization schemes \cite{BB23}.

Our theoretical results guarantee convergence of the expansion in a
neighbourhood of the expansion point, but the resulting rigorous bounds are
clearly non-sharp.  That is because our method of proof relies on the one hand on recursive estimates due to the structure of the linear Volterra equations and on the other hand on estimates of the integrated \textit{absolute} fractional-logarithmic kernels. For reliable use in option pricing, it is therefore useful to have a practical diagnostic for the effective radius of convergence. Since the expansion can be computed very efficiently for many values of \(H\), convergence can however be assessed empirically by comparing sufficiently high expansion orders
across Fourier nodes in the relevant pricing range and across the relevant
range of Hurst parameters. Both, the theoretical error estimate as well as the numerical results suggest that the radius of convergence generally depends on the size of the Fourier node in an inverse manner. Only for values sufficiently close to the expansion point, the small time asymptotics of the error estimate proposes that the series may converge for all Fourier arguments.

In the course of constructing the numerical benchmarks, we also record two auxiliary numerical variants for the direct approximation of the fractional Riccati equation; see Appendix~\ref{sec:implicit_riccati_pi} and Appendix~\ref{app:root-pade}. The first is a fully implicit  product integration scheme of the nonlinear Volterra equation. Since the Riccati driver is quadratic, the current-node update reduces to a scalar quadratic equation, whose branch is selected by comparison with the corresponding linearized implicit update. The second is a root-Pad\'e variant of the rational approximation of \cite{gatheral_radoicic_2019,gatheral_radoicic_2024}. Here the Pad\'e approximation is not applied directly to the Riccati solution, but to the transformed factors, defined through the stable and unstable stationary Riccati roots. These variants are not the main focus of the paper, but they provide significantly stable and accurate reference computations.

\textbf{Further related literature. }
Expansion methods are a standard tool in stochastic volatility modelling and appear in several forms depending on the parameter or regime being exploited. 
Classical examples include asymptotic smile formulae such as the SABR expansion \cite{HaganKumarLesniewskiWoodward02}, multiscale perturbation methods based on fast and slow volatility factors \cite{FouquePapanicolaouSircarSolna11}, martingale expansions for option prices and implied volatilities \cite{Fukasawa11}, expansions derived directly from characteristic functions \cite{JacquierLorig15}, second-order volatility expansions of implied volatility surfaces \cite{BergomiGuyon12}, Taylor expansions of the pricing operator in local-stochastic volatility models \cite{LorigPagliaraniPascucci13}, and algebraic expansions of the conditional characteristic function \cite{AlosGatheralRadoicic20,FrizGatheral25}.
For Heston-type models, small volatility-of-volatility expansions have been developed in \cite{BenhamouGobetMiri10}, and related perturbative methods have been used for extensions of the Heston model \cite{FouqueSaporito18}. In rough volatility, expansion and asymptotic methods appear in the analysis of
the short-time skew \cite{Fukasawa17}, small-time and large-deviation
asymptotics for rough stochastic volatility and Volterra models
\cite{FordeZhang17,JacquierPakkanenStone18,JacquierPannier22}, short-term
at-the-money expansions
\cite{ElEuchFukasawaGatheralRosenbaum19,JacquierMuguruzaPannier25}, refined
rough-volatility smile asymptotics
\cite{FrizGassiatPigato21,FrizGassiatPigato22}, and weak approximation
\cite{friz2025weak} or price expansions in forward variance models
\cite{BourgeyDeMarcoGobet22}. For the rough Heston model specifically, small-time, large-time, and $H\to0$
asymptotics have been studied in \cite{FordeGerholdSmith21}, while
\cite{AbiJaberDeCarvalho2024} study asymptotic limits on a reversionary time
scale.
To the best of our knowledge, expansions of option prices, characteristic functions, or other distributional quantities with respect to the Hurst parameter itself have not been studied in the rough volatility literature.
A natural direction for future work is to understand how such expansions in $H$ can be combined with existing expansions in other model parameters, option parameters, or asymptotic regimes.

The dependence on the kernel parameter in fractional and Volterra-type equations has been studied in several related settings, although mostly at the level of the driving noise process or the resulting stochastic equation rather than at the level of the Riccati equation considered here. In \cite{KN19}, it is shown that, for almost all $\omega\in\Omega$, the map $H\mapsto B^H_t(\omega)$ is infinitely often continuously differentiable and that the corresponding derivatives belong to $L^p$ for all $p\geq1$. In \cite{VGMU21}, continuity with respect to the Hurst parameter is established for rough differential equations driven by one-dimensional fractional Brownian motion for $H\in(1/3,1/2]$, in the sense of weak convergence with respect to the rough path topology. A broader overview of stability results for stochastic differential equations driven by fractional Brownian motion with respect to the Hurst parameter can be found in the recent survey \cite{KNP25}.

\textbf{Outline of the paper. }
Section~\ref{sec:riccati_expansion} studies the fractional Riccati equation as a parameter-dependent Volterra equation in $H$. We derive the recursive linear Volterra equations for the coefficients $\partial_H^n\psi(t;H_0,z)$ and prove the analyticity of $H\mapsto\psi(t;H,z)$ by solving a quadaratic recurrence relation. %
Section~\ref{sec:numerics} discusses the numerical implementation of the expansion and presents experiments for expansions around $H_0=1/2$ and $H_0=0$. The appendices collect details on the quadratic recurrence relations, the Volterra estimates, and the product-integration schemes used in the proofs and numerical implementation.

\section{Approximation of the fractional Riccati solution}\label{sec:riccati_expansion}

In the following, we will treat \eqref{eq:riccati} as a parameter depend Volterra equation, depending on $H\in(-1/2,\infty)$ and $z\in\C^*$. %
The main idea of the paper is to approximate $\psi$ by the first $N$ terms of its Taylor series in $H$ around $ H_0$, i.e.
\begin{equation}\label{eq:taylorpsi}
\psi(t;H,z)\approx\sum_{n=0}^N\partial_H^n\psi\left(t;H_0,z\right)\frac{(H-H_0)^n}{n!}.
\end{equation}
To this end, we will show that the partial derivatives with respect to $H$ do indeed exist and hence that the RHS of \eqref{eq:taylorpsi} is well-defined. Having computed $ \psi(\cdot;H_0,z)$, we formally differentiate \eqref{eq:riccati} with respect to $H$ to obtain an equation for $\partial_H\psi(\cdot;H_0,z)$: for all $H_0\in(-1/2,\infty)$, we have
\begin{equation}\label{eq:difH}
\begin{split}
    \partial_H\psi(t;H_0,z)&=\int_0^t\partial_HK^{H_0}(t-s) F\left(z,\psi(s;H_0,z)\right)ds\\
    &\quad+\int_0^tK^{H_0}(t-s)F_x\left(z,\psi(s;H_0,z)\right)\partial_H\psi(s;H_0,z)ds,
    \end{split}
\end{equation}
where 
\[\partial_HK^{H_0}(s)=K^{H_0}(s)\left(\ln(s)-\gamma(H_0+1/2)\right)\]
and $\gamma:=\Gamma'/\Gamma$ denotes the digamma function. The differentiability of $\psi$ with respect to $H$ and the validity of \eqref{eq:difH} is rigorously proven in Theorem \ref{thm:ndif} below. All further derivatives $\partial_H^n\psi(\cdot;H_0,z)$ for $n\geq 2,$ can be computed iteratively as solutions to similar linear Volterra equations, cf.~Theorem \ref{thm:ndif} below. 

Especially, for $H=1/2$ we note that the $0$-th order term in \eqref{eq:taylorpsi}, i.e.~${\psi}(t;1/2,z)$, solves the Riccati equation of the standard Heston model, i.e. 
\[{\psi}(t;1/2,z)=\int_0^t{F}\left(z,{\psi}(s;1/2,z)\right)ds,\quad t\in[0,T],\]
and is hence known in closed form. Similarly, the partial derivative $\partial_H\psi$ at $H=1/2$ solves the following standard linear first-order ODE: for all $t\in[0,T]$,
\begin{align*}
\partial_H\psi\left(t;1/2,z\right)&=\int_0^t F_x(z,\psi(s;1/2,z))\partial_H\psi\left(s;1/2,z\right)ds\\
&\qquad\qquad+\int_0^t(\ln(t-s)-\gamma(1)) F\left(z,\psi\left(s;1/2,z\right)\right)ds.
\end{align*}

In the following, we will provide a mathematically rigorous justification for our heuristic approach outlined above. 
It is well-known that %
for every fixed $H\in(-1/2,\infty)$ and $z\in\C^*$ there exists a unique continuous solution on $[0,T]$ to the fractional Riccati equation \eqref{eq:riccati}, cf.~\cite[Corollary 3.2]{EER18} and \cite[Theorem 4.1]{EER19}. In the next subsection, we will first show that the map $H\mapsto \psi(t;H,z)$ is infinitely often continuously differentiable and that the partial derivatives $\partial_H^n\psi(t;H,z),\ n\in\N,$ satisfy a certain linear first-order ODE. In a second step, we will prove in subsection \ref{subs:analyticity} the analyticity of $\psi(t;H,z)$ in $H$; more precisely, we will prove an error estimate for the remainder term of the associated Taylor series \eqref{eq:taylorpsi}. %
In what follows, we always set $T=1$. 

\subsection{Parameter dependence and differentiability of $\psi$ with respect to $H$}

In this subsection, we will prove that for every $z\in\C^*$ the unique solution to \eqref{eq:riccati} depends continuously on $H\in(-1/2,\infty)$ and is infinitely often continuously  differentiable with respect to $H$. 

\begin{The}\label{thm:ndif}
   The unique solution $\psi$ of \eqref{eq:riccati} is infinitely often continuously differentiable with respect to $H\in(-1/2,\infty)$ and its partial derivatives are the unique solution of the following linear Volterra equations: for $n\in\N$, $H\in(-1/2,\infty)$, and $z\in\C^*$,
   \begin{equation}\label{eq:ordern}
    \partial_H^n\psi(t;H,z)
        =\int_0^t K^{H}(t-s)F_x\left(z,\psi(s;H,z)\right)\partial^n_H\psi(s;H,z)ds+A_n(t;H,z),\quad t\in[0,1],
\end{equation} 
where 
\begin{equation}\label{eq:a_n}
\begin{split}
    A_n(t;H,z)&:=\int_0^t\partial_H^nK^H(t-s) F\left(z,\psi(s;H,z)\right)ds\\
       &\qquad+\sum_{k=1}^{n-1}\binom{n}{k}\int_0^t\partial_H^kK^H(t-s) F_x\left(z,\psi(s;H,z)\right)\partial^{n-k}_H\psi(s;H,z)ds\\
        &\qquad+\frac{\nu^2}{2}\sum_{k=0}^{n-2}\binom{n}{k}\int_0^t\partial^k_HK^H(t-s)\sum_{j=1}^{n-k-1}\binom{n-k}{j}\partial_H^j\psi(s;H,z)\partial^{n-k-j}_H\psi(s;H,z)ds
        \end{split}
\end{equation}
and 
    \begin{equation}\label{eq:diffkernel}
    \partial^n_HK^H(s)=B_n\left(\log(s)-\gamma(H+1/2),-\gamma^{(1)}(H+1/2),\dots,-\gamma^{(n-1)}(H+1/2)\right)K^H(s)
\end{equation}
with $\gamma^{(k)}$ denoting the $k$-th polygamma function and $B_n$ denoting the $n$-th complete Bell polynomial defined as
\begin{equation*}
    B_n(x_1,\dots,x_n):=\sum_{k_1,\dots,k_n\geq0:\ \sum_{i=1}^n ik_i=n}\binom{n}{k_1\dots k_n}\prod_{j=1}^n\left(\frac{x_j}{j!}\right)^{k_j}.
\end{equation*}
\end{The}

\begin{proof}
We want to apply Theorem \ref{thm:Voldif}. To this end, we may identify $\C$ with $\R^2$ and note that the two parametrized Volterra kernels \begin{align*}
   k(H,t,s)&:=\frac{(t-s)_+^{H-1/2}}{\Gamma(H+1/2)},\qquad s,t\in[0,1],\\
   \text{and}\quad k_H(H,t,s)&:=(\ln((t-s)_+)-\gamma(H+1/2))k(H,t,s),\qquad s,t\in[0,1],
\end{align*}
are both in $L^1([0,t],ds)$ for all $(t,H)\in[0,1]\times(-1/2,\infty)$ and 
satisfy for all $0\leq t_1\leq t_2\leq 1$
as $|t_2-t_1|\to0$,
\begin{align*}
\int_0^1|k(H,t_2,s)-k(H,t_1,s)|ds%
&=\frac{2(t_2-t_1)^{H+1/2}+t_1^{H+1/2}-t_2^{H+1/2}}{\Gamma(H+3/2)}\to0,
\end{align*}
and for $H\in(-1/2,1/2]$,
\begin{align*}
&\int_0^1|k_H(H,t_2,s)-k_H(H,t_1,s)|ds%
\leq \gamma(H+1/2)\int_0^1|k(H,t_2,s)-k(H,t_1,s)|ds\\
&+\frac{2(t_2-t_1)^{H+1/2}+t_1^{H+1/2}-t_2^{H+1/2}}{\Gamma(H+3/2)(H+1/2)}-\frac{2\ln(t_2-t_1)(t_2-t_1)^{H+1/2}+\ln(t_1)t_1^{H+1/2}-\ln(t_2)t_2^{H+1/2}}{\Gamma(H+3/2)},
\end{align*}
which also converges to zero as $|t_2-t_1|\to0$.
A similar continuity property of $k_H$ holds for $H>1/2$. Moreover, for all $t\in[0,1]$ and $H\in(-1/2,\infty)$,
\begin{align*}
    &\Gamma(H+1/2)\int_0^t|k(H+\varepsilon, t,s)-k(H,t,s)-\varepsilon k_H(H,t,s)|ds\\
    &\qquad\qquad=\int_0^ts^{H-1/2}\left(\frac{s^\varepsilon\Gamma(H+1/2)}{\Gamma(H+\varepsilon+1/2)}-1-\varepsilon \left(\log(s)-\gamma(H+1/2)\right)\right)ds\\
 &\qquad\qquad=t^{H+1/2}\left(\frac{t^{\varepsilon}\Gamma(H+1/2)}{\Gamma(H+\varepsilon+3/2)}-\frac{1-\varepsilon\gamma(H+1/2)}{H+1/2}
   +\frac{\varepsilon(1-(H+1/2)\log(t))}{(H+1/2)^2}\right)\\
   &\qquad\qquad\simeq\frac{(H+1/2)^2\frac{t^\varepsilon}{1+\varepsilon \gamma(H+1/2)}-(H+1/2)(H+\varepsilon+1/2)(1-\varepsilon\gamma(H+1/2))}{(H+1/2)^2(H+\varepsilon+1/2)}\\
      &\qquad\qquad\qquad +\frac{\varepsilon(H+\varepsilon+1/2)(1-(H+1/2)\log(t)}{(H+1/2)^2(H+\varepsilon+1/2)}\\
      &\qquad\qquad\simeq\frac{(1+\varepsilon\log(t))(1-\varepsilon\gamma(H+1/2))-1+\varepsilon\gamma(H+1/2)-\varepsilon\log(t)+o(\varepsilon)}{(H+\varepsilon+1/2)}  =o(\varepsilon).
\end{align*}
Since $ F\in C^\infty$, Theorem \ref{thm:Voldif} may be applied to the non-linear Volterra equation \eqref{eq:riccati}. Hence, for every $H\in(-1/2,\infty)$, there exists a unique solution to \eqref{eq:riccati} and the solution is continuously differentiable in $H$ with $\partial_H\psi$ satisfying \eqref{eq:ordern} for $n=1$. %

To prove the claim for higher order derivatives, we proceed by induction: suppose $\psi$ is continuously differentiable with respect to $H$ up to order $n$ and that \eqref{eq:ordern} holds for all $k\leq n$. Suppose that we can show that $A_n$ is continuously differentiable with respect to $H$ and that we may differentiate inside the integral. Then, we may apply Theorem \ref{thm:Voldif} to $\partial_H^n\psi(\cdot;\cdot,z)$ and may deduce that $\partial_H^{n+1}\psi(\cdot;\cdot,z)$ exists, is continuous, and satisfies for all $t\in[0,1]$ and $H\in(-1/2,\infty)$,
\begin{equation}\label{eq:n+1}
\begin{split}
   \partial_H^{n+1}\psi(t;H,z)=&\partial_HA_n(t;H,z)+\int_0^tK^H(t-s)F_x\left(z,\psi(s;H,z)\right)\partial_H^{n+1}\psi(s;H,z)ds\\
    &\quad +\int_0^t\partial_HK^H(t-s)F_x\left(z,\psi(s;H,z)\right)\partial_H^n\psi(s;H,z)ds\\
   &\quad+\nu^2\int_0^tK^H(t-s)\partial_H\psi(s;H,z)\partial_H^n\psi(s;H,z)ds.
   \end{split}
\end{equation}
If we are allowed to differentiate inside the integral, we will obtain 
\begin{align*}
&\partial_HA_n(t;H,z)=\int_0^t\partial_H^{n+1}K^H(t-s) F\left(z,\psi(s;H,z)\right)ds\\
&\quad+\int_0^t\partial_H^nK^H(t-s)\partial_H\psi(s;H,z) F_x\left(z,\psi(s;H,z)\right)ds\\
       &\quad+\sum_{k=1}^{n-1}\binom{n}{k}\int_0^t\partial_H^{k+1}K^H(t-s) F_x\left(z,\psi(s;H,z)\right)\partial^{n-k}_H\psi(s;H,z)ds\\
       &\quad+\sum_{k=1}^{n-1}\binom{n}{k}\int_0^t\partial_H^kK^H(t-s)F_x\left(z,\psi(s;H,z)\right)\partial^{n-k+1}_H\psi(s;H,z)ds\\
    &\quad+\nu^2\sum_{k=1}^{n-1}\binom{n}{k}\int_0^t\partial_H^kK^H(t-s)\partial_H\psi(s;H,z)\partial^{n-k}_H\psi(s;H,z)ds\\      
        &\quad+\frac{\nu^2}{2}\sum_{k=0}^{n-2}\binom{n}{k}\int_0^t\partial_H^{k+1}K^H(t-s)\sum_{j=1}^{n-k-1}\binom{n-k}{j}\partial_H^j\psi(s;H,z)\partial^{n-k-j}_H\psi(s;H,z)ds\\
        &\quad+\frac{\nu^2}{2}\sum_{k=0}^{n-2}\binom{n}{k}\int_0^t\partial_H^kK^H(t-s)\cdot\\
        &\quad\qquad\sum_{j=1}^{n-k-1}\binom{n-k}{j}\left[\partial_H^{j+1}\psi(s;H,z)\partial^{n-k-j}_H\psi(s;H,z)+\partial_H^j\psi(s;H,z)\partial^{n-k-j+1}_H\psi(s;H,z)\right]ds\\         
         &=\int_0^t\partial_H^{n+1}K^H(t-s) F\left(z,\psi(s;H,z)\right)ds\\
    &\quad+\sum_{k=2}^{n}\int_0^t\binom{n+1}{k}\partial_H^kK^H(t-s) F_x\left(z,\psi(s;H,z)\right)\partial^{n+1-k}_H\psi(s;H,z)ds\\      
  &\quad+\int_0^tn\partial_HK^H(t-s)F_x\left(z,\psi(s;H,z)\right)\partial^{n}_H\psi(s;H,z)ds\\    
        &\quad+\frac{\nu^2}{2}\sum_{k=0}^{n-1}\int_0^t\partial_H^kK^H(t-s)\binom{n+1}{k}\sum_{j=1}^{n-k}\binom{n+1-k}{j}\partial_H^j\psi(s;H,z)\partial^{n+1-k-j}_H\psi(s;H,z)ds\\
        &\quad-\nu^2\int_0^tK^H(t-s)\partial_H\psi(s;H,z)\partial_H^n\psi(s;H,z)ds.
\end{align*}
Plugging this expression into \eqref{eq:n+1} we get \eqref{eq:ordern} for $n+1$. Hence, it only remains to show that $A_n$ is continuously differentiable with respect to $H$, provided that $\psi$ is $n$ times continuously differentiable in $H$, and that we may differentiate inside the integral. But this follows easily from the fact that for all $s>0$ and $k\in\N$ by Fa\`a-di Bruno formula,
\begin{equation*}
    \partial^n_HK^H(s)=B_n\left(\log(s)-\gamma(H+1/2),-\gamma^{(1)}(H+1/2),\dots,-\gamma^{(n-1)}(H+1/2)\right)K^H(s),
\end{equation*}
and
\[(t-s)^{H-1/2}\ln^k(t-s)\in L^1([0,t],ds),\]
together with the continuity of the functions $F$ and $F_x$ and the induction hypothesis, which in particular implies that $\psi$ and its first $n$ partial derivatives $\partial_H^k\psi(\cdot;\cdot,z),\ k=1,\dots,n$, are continuous on $[0,1]\times (-1/2,\infty)$ for every $z\in\C^*$. 
\end{proof}

\subsection{Analyticity of $\psi$ in $H$}\label{subs:analyticity}

We want to approximate $\psi(t;H,z)$ by its Taylor polynomial of degree $N$ around $H_0$, i.e.~by
\begin{equation}\label{eq:taylor}
P_{N,H_0}(t;H,z):=\sum_{n=0}^N\partial_H^n\psi\left(t; H_0,z\right)\frac{(H-H_0)^n}{n!}.
\end{equation}
To estimate the remainder of order $N$ of the approximation of $\psi(t;H,z)$ through $P_{N,H_0}(t;H,z)$, i.e.
\begin{equation*}
R_{N,H_0}(t;H,z):=\psi(t;H,z)-P_{N,H_0}(t;H,z),
\end{equation*}
we can apply Taylor's theorem separately to the real and imaginary part of $\psi(t;H,z)$ and obtain the upper bound
\begin{equation}\label{eq:restglied}
\left|R_{N,H_0}(t;H,z)\right|\leq \sqrt{2} \cdot\frac{\sup_{\tilde{H}\in[H\wedge H_0,H\vee H_0]}\left|\partial_H^{N+1}\psi\left(t;\tilde{H},z\right)\right|}{(N+1)!}\cdot |H-H_0|^{N+1}.
\end{equation}
Recall from Theorem \ref{thm:ndif} that
\begin{align*}
    \partial_H^n\psi(t;H,z)
        =\int_0^t (t-s)^{H-1/2}F_x(z,\psi(s;H,z))\partial^n_H\psi(s;H,z)ds+A_n(t;H,z),
\end{align*}
where $A_n$ was defined in \eqref{eq:a_n}. To estimate the Taylor remainder we have to divide the above Volterra equation by $n!$ on both sides and to find a bound for
\begin{equation}\label{eq:an}
\begin{split}
    a_n(t;H,z)&:=\frac{A_n(t;H,z)}{n!}=\int_0^t\frac{K^H_n(t-s)}{n!} F\left(z,\psi(s;H,z)\right)ds\\
        &+\sum_{k=1}^{n-1}\int_0^t\frac{\partial_H^kK^H(t-s)}{k!}
        F_x(z,\psi(s;H,z))\frac{\partial^{n-k}_H\psi(s;H,z)}{(n-k)!}ds\\
        &+\frac{\nu^2}{2}\sum_{k=0}^{n-2}\int_0^t\frac{\partial_H^kK^H(t-s)}{k!}\sum_{j=1}^{n-k-1}\frac{\partial_H^j\psi(s;H,z)}{j!}\frac{\partial^{n-k-j}_H\psi(s;H,z)}{(n-k-j)!}ds.
\end{split}
\end{equation}
This will be the main challenge in order to prove the analyticity of $\psi$ in $H\in(-1/2,\infty)$. We have the following result, which 
shows that the Taylor series for $\psi(t;H,z)$ in any $H_0\in(-1/2,\infty)$ converges uniformly to $\psi(t;H,z)$ in a neighbourhood of $H_0$, i.e.~the function $H\mapsto\psi(t;H,z)$ is real analytic for every $z\in\C^*$ and $t\in[0,1]$.

\begin{The}\label{thm:errorbound}
Let $t\in[0,1]$, $z=a+ib\in \C^*$ and $H,H_0\in(-1/2,0.95)$. Then for all $N\in\N$,
\begin{align*}
    \sup_{t\in[0,1]}\left|R_{N,H_0}(t;H,z)\right|&\leq \frac{1}{\sqrt{2}N}\left(\frac{H\wedge H_0}{2}+\frac{1}{4}\right)^{-N-1}\left|H-H_0\right|^{N+1}\cdot \\
&\quad\left(\left(r_1(b,H\wedge H_0)+r_2(b,H\wedge H_0)\right)^2+\frac{r^2_1(b,H\wedge H_0)}{2}\right)^{(N+1)/2},
\end{align*}
where for $|b|\to\infty$,
\[r_1(b,H\wedge H_0)\sim_{H\wedge H_0} 1+|b|^4\qquad\text{and}\qquad r_2(b,H\wedge H_0)\sim_{H\wedge H_0}1+|b|^8.\]
Especially, the function $H\mapsto\psi(t;H,z)$ is real analytic on $(-1/2,0.9)$ for every $z\in\C^*$ and $t\in[0,1]$.\footnote{With slightly more effort, one could also prove the analyticity on the whole interval $(-1/2,\infty)$ instead of $(-1/2,0.95)$. However, we are only interested in the rough case $H<1/2$ and $H_0\leq 1/2$, so we simply used the fact that the digamma function is strictly negative and does not cross zero on the interval $(0,1.45)$ to derive our error estimate.}
\end{The}

It follows from Theorem \ref{thm:errorbound} and the identity theorem for real analytic functions that if the Taylor series converges for given $H,H_0,z$, then we have the equality $P_{\infty,H_0}(t;H,z)=\psi(t;H_0,z)$ for all $t\in[0,1]$. Hence, we can use the Taylor polynomial $P_{N,H_0}$ to approximate $\psi$. 

On the other hand, the estimate in Theorem \ref{thm:errorbound} is clearly suboptimal and does therefore not allow us to compute the convergence radius. The reason is that we rely on the recursive structure of the Volterra equations to derive the error bound and that we pull the absolute value inside the integral. This is suboptimal, since $\C\ni H\mapsto \int_0^t \frac{|\partial_H^nK^H(s)|}{n!}ds$ is not an entire function, even though for every $s>0$, the function $\C\ni H\mapsto \partial _H^nK^H(s)$ is entire and has thus infinite radius of convergence. However, as we will see in Section \ref{sec:numerics}, the radius of convergence can be analyzed numerically, cf.~Figure \ref{fig:ratio_heatmap}. We note that both, the theoretical error estimate as well as the numerical results in Section \ref{sec:numerics} reveal an inverse dependence of the convergence radius on $|b|$.

\begin{proof} The proof of the error estimate will be split in three steps: First, we apply a fractional Gronwall inequality to reduce the problem to the estimation of the inhomogeneous part. Second, we bound the fractional logarithmic kernels and third, we solve a quadratic recurrence equation. 
\begin{enumerate}[leftmargin=*]
    \item
By Theorem \ref{thm:ndif}, the function $$\psi_n(\cdot;H,z):=\frac{\partial^n_H\psi(\cdot;H,z)}{n!}$$ is the unique, continuous solution to the linear Volterra equation 
   \begin{equation*}%
    \psi_n(t;H,z)
        =\int_0^t (t-s)^{H-1/2}F_x(z,\psi(s;H,z))\psi_n(s;H,z)ds+a_n(t;H,z),\quad t\in[0,1].
\end{equation*} 
It follows from Lemma \ref{lem:boundpsi} that $\Re(\psi(s;H,z))\leq 0$ for all $s,H,z$ and hence  also
\[\Re\left(F_x(z,\psi(s;H,z))\right)\leq \nu-\lambda+\nu^2\Re\left(\psi(s;H,z)\right)\leq 0.\]
Therefore, Lemma \ref{lem:linearvolterrabound} implies that for all $t\in[0,1]$, 
\begin{equation}\label{eq:apriori}
\begin{split}
\left|\psi_n(t;H,z)\right|&\leq \left|a_n(t;H,z)\right|+\sup_{s\leq t}\left|a_n(s;H,z)F_x(z,\psi(s;H,z))\right|\int_0^tE_0(s)ds\\
   &\leq \left(1+\frac{t^{H+1/2}}{\Gamma(H+3/2)}\sup_{s\leq t}|F_x(z,\psi(s;H,z))|\right)\sup_{s\leq t}|a_n(s;H,z)|,
   \end{split}
\end{equation}
where $E_0$ is the canonical resolvent of $K^H$ with parameter $0$, which actually equals $K^H$ itself, i.e.~$E_0=K^H$.
Let us denote for $b\in\R$,
\begin{align*}
\beta_H(b)&:=|\rho|\nu (1+|b|)+\lambda+\nu^2 C_{1,H}\left(\frac{1}{2}+b^2\right),\\
\gamma_H(b)&:=\left(\frac{1}{2}+(\lambda+|\rho|\nu(1+|b|)) C_{1,H}\right)\left(\frac{1}{2}+b^2\right)+\frac{\nu^2}{2}C_{1,H}^2\left(\frac{1}{2}+b^2\right)^2,
\end{align*}
where $C_{1,H}$ is the constant from Lemma \ref{lem:boundpsi}.  
Then for all $H\in(-1/2,\infty)$ and $b\in\R$, %
\[\sup_{s\in[0,1]}\left| F\left(z,\psi(s;H,z)\right)\right|\leq\gamma_H(b),\qquad \sup_{s\in[0,1]}\left|F_x(z,\psi(s;H,z))\right|\leq\beta_H(b).\]
Denoting for all $n\in\N$, 
\begin{align*}
c_n(H,z)&:=\sup_{t\in[0,1]}|\psi_n(t;H,z)|,
\end{align*}
we obtain from \eqref{eq:an} that for all  $t\in[0,1]$,
\begin{align*}
|a_n(t;H,z)|&\leq \gamma_H(b)\int_0^t\frac{|\partial_H^nK^H(t-s)|}{n!}ds+\beta_H(b)\sum_{k=1}^{n-1}c_{n-k}(H,z)\int_0^t\frac{|\partial_H^kK^H(t-s)|}{k!}ds\\
&\qquad+\frac{\nu^2}{2}\sum_{k=0}^{n-2}\int_0^t\frac{|\partial_H^kK^H(t-s)|}{k!}ds\sum_{j=1}^{n-k-1}c_j(H,z)c_{n-k-j}(H,z).
\end{align*}

\item Next, we estimate for all $n\in\N$ the integral
\begin{align*}
I_n(t)&:=\int_0^t\left|\partial^n_HK^H(s)\right|ds\\
&=\int_0^t\left|B_n\left(\log(s)-\gamma(H+1/2),-\gamma^{(1)}(H+1/2),\dots,-\gamma^{(n-1)}(H+1/2)\right)\right|K^H(s)ds.    
\end{align*}
By definition of the $n$-th complete Bell polynomial, we have for all $x_1,\dots,x_n\in\R$ and $a\geq 1$,
\[|B_n(x_1,\dots,x_n)|\leq B_n(|x_1|,\dots,|x_n|)\leq a^nB_n\left(|x_1|,\frac{|x_2|}{a^2},\dots,\frac{|x_n|}{a^n}\right).\]
Moreover, if $|x_k|\leq k!\delta^k|x_1|$ with $\delta|x_1|<1$, then we have the estimate
\[B_n(|x_1|,\dots,|x_n|)\leq (1+|x_1|)^n.\]
Since
\[\gamma^{(k)}(H+1/2)=(-1)^{k+1}k!\cdot \zeta(k+1,H+1/2),\]
where $\zeta$ denotes the Hurwitz zeta function, 
we have for $k\geq 2$ the bound
\[|\gamma^{(k)}(H+1/2)|\leq k! \left(\frac{1}{2(H+1/2)^{k+1}}+\frac{1}{k(H+1/2)^k}\right)\leq \frac{k!}{(H+1/2)^{k+1}}.\]
Note that $\gamma(H+1/2)<-1/2$ for all $H\in(-1/2,0.95)$. Choosing $a=\delta^{-1}(H+1/2)^{-1}$ for some $\delta\in(0,H+1/2)$, we get
\begin{align*}
    \frac{|\gamma^{(k-1)}(H+1/2)|}{a^k}\leq (k-1)!\delta ^k\leq -k!\delta^k(\log(s)+\gamma(H+1/2))
\end{align*}
for all $s\in(0,1]$ and thus
\begin{align*}
    &\left|B_n\left(\log(s)-\gamma(H+1/2),-\gamma^{(1)}(H+1/2),\dots,-\gamma^{(n-1)}(H+1/2)\right)\right|\\
    &\qquad \qquad\leq \delta^{-n}(H+1/2)^{-n}(1-\log(s)-\gamma(H+1/2))^n.
\end{align*}
Now, we have for all $t>0$, 
\begin{align*}
\int_0^ts^{H-1/2}(1-\log(s))^nds%
&=\frac{e^{H+1/2}}{(H+1/2)^{n+1}}\Gamma\left(n+1,(H+1/2)(1-\gamma(H+1/2)-\log(t))\right).%
\end{align*}
Since $\Gamma(n+1,x)\leq\Gamma(n+1)=n!$ for all $x>0$, we obtain for all $t\in[0,1]$
\begin{align*}
|a_n(t;H,z)|&\leq\frac{e^{H+1/2}\gamma_H(b)}{(H+1/2)^{n+1}}+e^{H+1/2}\beta_H(b)\sum_{k=1}^{n-1}\frac{c_{n-k}(H,z)}{(H+1/2)^{k+1}}\\
&\qquad\qquad\qquad+e^{H+1/2}\frac{\nu^2}{2}\sum_{k=0}^{n-2}\sum_{j=1}^{n-k-1}\frac{c_j(H,z)c_{n-k-j}(H,z)}{(H+1/2)^{k+1}}.
\end{align*}
Plugging this into the a-priori estimate \eqref{eq:apriori} and multiplying both sides with $(H+1/2)^n$, we obtain the recursive inequality
\begin{align*}
c_n(H,z)(H+1/2)^n&\leq
 K(b)\left[\gamma_H(b)+\beta_H(b)\sum_{k=1}^{n-1}c_{k}(H,z)(H+1/2)^{k}\right.\\
 &\qquad\qquad\qquad\left.+\frac{\nu^2}{2}\sum_{k=2}^n\sum_{i,j\geq1:\ i+j=k}c_i(H,z)c_j(H,z)(H+1/2)^{i+j}\right]
\end{align*}
with
\[K_H(b):=1+\frac{\beta_H(b)}{\Gamma(H+3/2)}.\]

\item For ease of notation, we suppress the dependence on $H$ and $z=a+ib$ in the following. %
The last inequality implies that for all $n\in\N$,
$$c_n\leq c'_n\left(H+\frac{1}{2}\right)^{-n},$$ 
where the series $(c'_n)_{n\in\N}$ satisfies 
\begin{align*}
c'_1&:=K\gamma,\\
    c'_n&:=K\gamma+K\beta\sum_{k=1}^{n-1}c'_{k}+\frac{K\nu^2}{2}\sum_{i,j\geq1:\ i+j\leq n} c'_ic'_j\qquad\text{for}\qquad n\geq2.
\end{align*}
Therefore
\[c_n'-c_{n-1}'=K\left(\beta c'_{n-1}+\frac{\nu^2}{2}\sum_{i=1}^{n-1}c'_ic'_{n-i}\right),\quad n\geq 2,\]
and $(c'_n)_{n\in\N}$ follows the quadratic recursion
\[c_n'=\alpha_1c'_{n-1}+\alpha_2\sum_{i=1}^{n-1}c'_ic'_{n-i},\quad n\geq 2,\]
with $$\alpha_1:=1+K\beta\qquad\text{and}\qquad\alpha_2:=\frac{K\nu^2}{2}.$$ 
Applying Lemma \ref{lem:recurrence} we obtain for $n\geq2$ that 
\begin{align*}
    c_n\leq \frac{2^n(\alpha_1+2K\alpha_2\gamma)^n}{2(n-1)}\cdot\left(1+\frac{\alpha_1^2}{2(\alpha_1+2K\alpha_2\gamma)^2}\right)^{n/2}\left(H+\frac{1}{2}\right)^{-n},
\end{align*}
which implies for all $H,H_0\in(-1/2,0.95)$ and $n\geq2$,
\begin{align*}
&\sup_{t\in[0,1],\tilde{H}\in[H\wedge H_0,H\vee H_0]}\left|\partial_H^{N+1}\psi_n\left(t;\tilde{H},z\right)\right|\cdot |H-H_0|^{N+1}\\
&\leq \frac{1}{2N}\left(\frac{H\wedge H_0}{2}+\frac{1}{4}\right)^{-N-1}\left|H-H_0\right|^{N+1}\cdot \\
&\quad\left((1+K_{H\wedge H_0}(b)\beta_{H\wedge H_0}(b)+\nu^2K^2_{H\wedge H_0}(b)\gamma_{H\wedge H_0}(b))^2+\frac{(1+K_{H\wedge H_0}(b)\beta_{H\wedge H_0}(b))^2}{2}\right)^{(N+1)/2}.
\end{align*}
Thus, the error estimate follows from \eqref{eq:restglied}.
\end{enumerate}
\end{proof}

\begin{Cor}
For every $z\in\C^*$, the characteristic function $\varphi^H(z)$ of the log-stock price $X_T$ is real analytic in $H$ on $(-1/2,0.95)$.    
\end{Cor}

\begin{proof}
As the exponential function is analytic and the decomposition of analytic functions is again analytic, it suffices to prove that 
\[H\mapsto\int_0^TF(z,\psi(T-t;H,z))g^H(t)dt\]
is real analytic. Because of Theorem \ref{thm:errorbound}, this will be the case if we are allowed to interchange integration and differentiation. However, this follows from the dominated convergence theorem as $\left|\partial_H^n\psi(s;H,z)\right|$ is continuous in $H$ and hence uniformly bounded by an integrable function on any compact subinterval of $(-1/2,0.95)$ as shown in step (2) in the proof of Theorem \ref{thm:errorbound}.
\end{proof}

In Theorem \ref{thm:errorbound} we derived an error bound uniformly for all $t\in[0,1]$. If one keeps track of $t$ throughout the derivation, one can study the small time asymptotics as $t\to0$. This is the content of the following corollary. 

\begin{Cor}\label{cor:smalltime}
Let $z=a+ib\in \C^*$ and $H,H_0\in(-1/2,\infty)$. Then for all $N\in\N$ and $M>0$,
\begin{align*}
    \lim_{t\to0}\sup_{z\in\C^*_M}\left|R_{N,H_0}(t;H,z)\right|\leq  \frac{1}{\sqrt{2}N}\left(\frac{3}{2}\right)^{(N+1)/2}\left(\frac{H\wedge H_0}{4} +\frac{1}{8}\right)^{-N-1} \left|H-H_0\right|^{N+1}.
\end{align*}
Hence, for any $M>0$, $H_0\in(-1/2,\infty)$, and $H\in(-1/2,\infty)$ satisfying
\[|H-H_0|<\sqrt{\frac{2}{3}}\left(\frac{H_0}{4}+\frac{1}{8}\right),\]
there exists $t_0>0$ such that $P_{\infty,H_0}(t;H,z)=\psi(t;H,z)$ for all $t\in[0,t_0)$ and $z=a+ib\in \C^*_M$, where $\C^*_M:=\{z=a+ib\in\C^*:|b|\leq M\}.$ 
\end{Cor}

\begin{proof}
The proof follows along the same lines as the proof of Theorem \ref{thm:errorbound}, but we consider the $t$-dependent quantity
\[c_n(t;H,z):=\sup_{s\in[0,t]}|\psi_n(s;H,z)|\]
and replace the coarse, uniform estimate $\Gamma(k,x)/\Gamma(k)\leq 1$ in step (2) by the following $x$-dependent estimate:
for all $n\in\N$ and $x>0$,
\begin{align*}
\Gamma(k,x)=\int_x^\infty u^{k-1}e^{-u}du=e^{-x}\int_0^\infty(u+x)^{k-1}e^{-u}du\leq 2^{k-2}\left[\Gamma(k)+x^{k-1}\right]e^{-x}
\end{align*}
and therefore
\[\frac{\Gamma(k,x)}{\Gamma(k)}\leq 2^{k-2}\left[\frac{x^{k-1}}{(k-1)!}+1\right]e^{-x}\leq 2^{k-2}e^{-x}\cdot \sup_{k\in\N}\left[\frac{x^{k-1}}{(k-1)!}+1\right]=:D(x)2^{k+1},\]
where $D(x)\to0$ as $x\to\infty$ because
\[\sup_{k\in\N}\frac{x^k}{k!}=\frac{e^x}{\sqrt{2\pi x}}\left(1+\mathcal{O}\left(\frac{1}{x}\right)\right).\]
Denoting $\delta(H,t):=D(-(H+1/2)\ln(t))$, we have for all $n\geq 2$,
$$c_n\leq c'_n\left(\frac{H}{2}+\frac{1}{4}\right)^{-n},$$ 
where
\[c_n'=\alpha_1c'_{n-1}+\alpha_2\sum_{i=1}^{n-1}c'_ic'_{n-i},\quad n\geq 2,\]
with $$\alpha_1:=1+2K\delta\beta,\qquad\alpha_2:=K\delta\nu^2,\qquad c_1'=2K\gamma\delta.$$ 
Hence, we obtain for $n\geq2$ that 
\begin{align*}
    c_n\leq \frac{2^n(\alpha_1+4K\alpha_2\gamma\delta)^n}{2(n-1)}\cdot\left(1+\frac{\alpha_1^2}{2(\alpha_1+4K\alpha_2\gamma\delta)^2}\right)^{n/2}\cdot\left(\frac{H}{2}+\frac{1}{4}\right)^{-n}.
\end{align*}
Now the claimed asymptotics follows by observing that $\alpha_1\to1$ and $\alpha_2\to0$ as $t\to0$, uniformly in all $z\in\C^*_M$. Also note that as we consider the case $t\to 0$, we do not need a lower positive bound on $\gamma(H+1/2)$. Hence, the result holds for all $H,H_0\in(-1/2,\infty)$.

Finally, the Taylor series will converge on a small enough time horizon if 
\[\sqrt{\frac{3}{2}}|H-H_0|\leq \frac{H\wedge H_0}{4}+\frac{1}{8},\]
which is the case if and only if
\[-\left(\sqrt{\frac 3 2}+\frac 1 4\right)^{-1}\left(\frac{H_0}{4}+\frac{1}{8}\right)<H-H_0<\sqrt{\frac{2}{3}}\left(\frac{H_0}{4}+\frac{1}{8}\right).\]
As the interval of convergence must be symmetric around $H_0$, the claim follows. 
\end{proof}

Especially interesting is the choice $H_0=1/2$ as in this case the Volterra equations \eqref{eq:ordern} turn into standard linear ODEs. The above corollary implies that for $H_0=1/2$, the Taylor series $P_{\infty,1/2}(t;\cdot,z)$ has a convergence radius of at least $0.2$ for all $z\in\C^*$, if $t$ is small enough, where $t\to0$ as $|b|\to\infty$. Also for $T$ not converging to zero, the error bound of Theorem \ref{thm:errorbound} can be improved in the case $H_0=1/2$, since we have a better a priori bound on the solution of the standard Heston Riccati equation than on the solution of the fractional Riccati equation. 

\begin{Cor}
    For all $H\in(-1/2,\infty)$, $z=a+ib\in\C^*$, and $N\in\N$, we have
    \begin{align*}
    \sup_{t\in[0,1]}\left|R_{N,1/2}(t;H,z)\right|\leq \frac{2^{N+1}}{\sqrt{2}N}\left|H-H_0\right|^{N+1}\left(\left(r_1(b,1/2)+r_2(b,1/2)\right)^2+\frac{r^2_1(b,1/2)}{2}\right)^{(N+1)/2},
\end{align*}
where for $|b|\to\infty$,
\[r_1(b,1/2)\sim 1+|b|^2\qquad\text{and}\qquad r_2(b,1/2)\sim1+|b|^4.\]
\end{Cor}

\begin{proof}
    The proof follows along the same lines as the proof of Theorem \ref{thm:errorbound}, but uses the a priori bound established in Lemma \ref{lem:boundHeston}. Also, since we already know that the function is real analytic, it suffices to estimate $|\psi_n(t;H_0,z)|$ instead of $\sup_{\tilde H}|\psi_n(t;\tilde{H},z)|$. Hence, in this case we get the same recurrence relation, but with   
\begin{align*}
\beta_{1/2}(b)&:=|\rho|\nu (1+|b|)+\lambda+\nu^2 C\left(1+|b|\right),\\
\gamma_{1/2}(b)&:=\frac{b^2+|b|+1/4}{2}+(\lambda+|\rho|\nu(1+|b|))C\left(1+|b|\right)+\frac{\nu^2}{2}C^2\left(1+|b|\right)^2,
\end{align*}
where $C$ is now the constant from Lemma \ref{lem:boundHeston}.  
\end{proof}

\section{Numerical discussion}\label{sec:numerics}

We discuss the expansion numerically and test it on the pricing of European call options. To this end, we first describe the implementation of the expansion and the experimental setup used in the examples. The full code for the numerical experiments is available at \url{https://github.com/hagerpa/roughprix}.

\subsection{Implementation}\label{sec:implementation}{\ }

\textbf{Reference solution. }
We need numerical reference solutions to the fractional Riccati equation \eqref{eq:riccati} for different values of $H$ for two reasons:
Firstly, the expansion starts from a known solution $\psi^{H_0}$. 
Secondly, we need to compute benchmark values in order to compare the accuracy of our approximation.

The standard approach for approximating a solution to \eqref{eq:riccati} is to discretize the Volterra integral in time, often by a fractional Adams--Bashforth--Moulton predictor--corrector method; see, e.g., \cite{diethelm_ford_freed_2002,EER19}. As frequently reported in the literature, accurate option prices may require very fine time grids, often with thousands or even tens of thousands of time steps, especially for short maturities and large Fourier frequencies. Since the computational cost scales quadratically in the number of time steps, this quickly becomes infeasible, and several approaches have been proposed to accelerate the computation; see also the Introduction~\ref{sec:intro}.
For the benchmark computations in this paper, we found an implicit product-integration scheme described below most suitable.

While it is still within the spirit of our approach to compute a possibly expensive but accurate starting solution at $H_0$ and then obtain approximate solutions efficiently for nearby values of $H$ by expansion, the two expansion points used below admit more direct starting solutions. For $H_0=1/2$, the starting solution is the explicit classical Heston Riccati solution. For $H_0=0$, we instead use an accurate closed-form Pad\'e approximation. We describe both choices further below.

\textit{Benchmark scheme.} To compute independent benchmarks for validating the expansion results, we use a direct discretization of the fractional Riccati equation. Since the Adams-type scheme requires very fine time discretizations for $H<0$, as also reported in the literature \cite{BB23}, we use instead the simple product-integration scheme described in Section~\ref{sec:implicit_riccati_pi}. Compared with hybrid schemes such as \cite{callegaro_grasselli_pages_2020}, this remains a plain discretization of the fractional Riccati equation. The scheme treats the Riccati nonlinearity implicitly, and the resulting quadratic update is resolved by a continuous selection of the appropriate Riccati root. With this discretization, $2\cdot 10^3$ time steps were sufficient to obtain stable benchmark values for all Hurst parameters considered in the experiments, ranging from $H=-0.3$ to $H=0.4$.

\textit{Classic Heston.}
Firstly, the case $H_0 = 1/2$ reduces to the classical Heston case, and the solution of the Riccati equation \eqref{eq:riccati} is available in closed form.
As is well known, care needs to be taken in the precise expression due to the choice of branches of the complex square root and the possible discontinuities in equivalent formulations of the Heston characteristic function, often referred to as the Little Heston Trap; see e.g.~\cite{KL06,Albrecher}. We use
\[
  \psi^{1/2}(t,z)
    =
    \frac{\beta(z)-d(z)}{\nu^2}
    \frac{1-e^{-d(z)t}}{1-g(z)e^{-d(z)t}},
\]
where
\[
  d(z)=\sqrt{\beta^2(z)-\nu^2(z^2-z)},
  \qquad
  \beta(z):=\lambda-\rho \nu z,
  \qquad
  g(z):=\frac{\beta(z)-d(z)}{\beta(z)+d(z)}
\]
and the principal branch of the square root is chosen, i.e.~$\Re(d(z))\geq0$; cf.~\cite[Theorem 3]{KL06}.

The expansion further simplifies here, as the linear Volterra equations \eqref{eq:ordern} reduce to classical linear ODEs.
Moreover, the homogeneous solution is available explicitly:
\[
  \exp\left(
    \int_0^t
     F_x
    \left(z,\psi^{1/2}(s,z)\right)
    \dd{s}
  \right)
  =
  \frac{
    e^{-d(z)t}\big(1-g(z)\big)^2
  }{
    \big(1-g(z)e^{-d(z)t}\big)^2
  }.
\]
Thus the linearized equations can be solved by the usual variation-of-constants formula.

\textit{Pad\'e approximation.} Secondly, the fractional Riccati equation can also be treated by rational approximation methods. Gatheral and Radoi\v{c}i\'c \cite{gatheral_radoicic_2019,gatheral_radoicic_2024} propose Pad\'e-type approximations of the rough Heston Riccati solution, showing that such rational approximations can yield fast and highly accurate option prices in the relevant Fourier-pricing domain. In our numerical experiments for expansions around $H_0=0$, we use the root-Pad\'e variant of degree $m=7$ described in Appendix~\ref{app:root-pade}, which improves the approximation in the out-of-the-money region.

\textbf{Expansion. }
The implementation of the Riccati expansion itself is straightforward, as all combinatorial terms have been laid out in the theory above, and the computation of the expansion terms essentially amounts to solving linear Volterra equations with fractional-logarithmic kernels.
We discretize these equations by a piecewise-linear product-integration method detailed in Section~\ref{sec:fraclogvolterra}.

As already mentioned above, in the classical Heston case $H_0 = 1/2$, the situation simplifies considerably.
The kernel without logarithmic derivatives is regular, and the linear Volterra equations reduce to standard inhomogeneous ODEs with the same homogeneous part for all expansion orders.
Since this homogeneous solution is available explicitly, the corresponding linearized equations are implemented directly by a variation-of-constants formula.

\textbf{Fourier pricing. }
In the numerical experiments we consider the valuation of European call options with payoff
$f_K(x)=(e^x-K)^+$, $K>0$, written as a function of the log-price
$X_T^H=\log(S_T^H/S_0^H)$. 
By the Lewis Fourier pricing formula \cite{Lewis2001,GatheralBook}, and with $S_0^H=1$,
\begin{equation}\label{eq:call_price_fourier}
    C^{H}(T,K) := \mathbb E\big[f_K(X_T^H)\big]
    =
    1-\frac{\sqrt K}{\pi}
    \int_0^\infty
    \Re\left(
        \frac{
            e^{iu\log K}\varphi^{H}(1/2-iu)
        }{
            u^2+1/4
        }
    \right)\dd u,
\end{equation}
where $\varphi^H$ is the characteristic function of $X_T^H$ defined in \eqref{eq:charfct}.
Numerically, we also tested the standard exponentially damped vertical contours
$z=R-iu$, $R>1$; see \cite{EGP10}. In the parameter regimes
considered below, these tests led to the same qualitative conclusions. These
contours, however, lie outside the strip on which our expansion estimates are
stated, and are therefore not directly covered by the convergence theory
developed above.

Using the expansion of the Riccati equation from Section~\ref{sec:riccati_expansion}, we replace $\varphi^H$ in \eqref{eq:call_price_fourier} by
\[
    \varphi^{H}_{N,H_0}(z)
    :=
    \exp\left(\int_0^T
    F\left(z,P_{N,H_0}(T-t;H,z)\right)g^H(t)\dd t\right),
\]
around an expansion point $H_0\in(-1/2,1/2]$, where $P_{N,H_0}$ is defined in \eqref{eq:taylor}, and $g^H$ and $F$ are defined in \eqref{def:gH} and \eqref{def:F}. The contour cutoff $M$ is chosen adaptively from a predefined set of
integration endpoints $0=M_0<M_1<\dots<M_{P_{\mathrm{max}}}$, increasing the Fourier domain
until the computed prices, or the corresponding implied volatilities, change by
less than the prescribed tolerance.

Assuming that the reference Riccati solution at the expansion point $H_0$ can
be precomputed accurately once, we use the decomposition
\[
    \varphi^{H}_{N,H_0}(z)
    =
    \varphi^{H_0}(z)\,
    m^{H}_{N,H_0}(z).
\]
We then build the Fourier quadrature weights from the product of the Lewis
kernel and the reference characteristic function $\varphi^{H_0}$. On each interval
$[M_{p-1},M_p]$, let $(u_{p,j})$ denote the coarser grid on which the
correction factor is evaluated, and let $(\ell_{p,j})$ be the corresponding
local interpolation functions. In the implementation we use Chebyshev--Lobatto nodes for $(u_{p,j})$
and construct the weights by Gauss--Legendre quadrature on each interval. We
precompute weights of the form
\[
    W^{H_0}_{K,p,j}
    :=
    \int_{M_{p-1}}^{M_p}
    \frac{
        e^{iu\log K}
        \varphi^{H_0}(1/2-iu)
    }{
        u^2+1/4
    }
    \ell_{p,j}(u)\dd u .
\]
The expansion then enters only through the values
$m^{H}_{N,H_0}(1/2-iu_{p,j})$, which are combined with these precomputed
weights:
\begin{equation}\label{eq:call_approx}
        C^{H}(T,K)
    \approx
    1-\frac{\sqrt K}{\pi}
    \Re\left(
        \sum_{p=1}^P \sum_j
        W^{H_0}_{K,p,j}\,
        m^{H}_{N,H_0}(1/2-iu_{p,j})
    \right).
\end{equation}
This separates the accurate computation of the reference solution from the
comparatively cheap computation of the expansion correction. In particular,
once the weights associated with $\varphi^{H_0}$ have been computed,
changing the expansion order only requires recomputing the multiplier
$m^{H}_{N,H_0}$ on the coarser grid.
\subsection{Experimental setup}

\textbf{Model Parameters. }
For the underlying Heston model we fix the following parameter set, commonly
used in the literature:
\[
    \lambda = 0.3,\qquad
    \theta = 0.3\cdot 0.02,\qquad
    \nu = 0.3,\qquad
    \rho = -0.7,\qquad
    V_0 = 0.02.
\]
We normalize the spot to $S_0=1$ and set the risk-free rate to $r=0$. The Hurst
parameter varies across the experiments. 
To keep the discussion tractable, we only document the expansions around the two
boundary cases $H_0=1/2$ and $H_0=0$. The implementation, however, allows
for a generic expansion point $H_0 \in (-1/2, 1/2]$, and the results interpolate analogously.

\textbf{Diagnostics. }
\textit{Implied volatilities.} As the main diagnostic quantities for the numerical efficiency of the expansion
we use Black--Scholes implied volatilities, as a scale for comparing
option-price errors across strikes and maturities and, in particular, for
detecting inaccuracies in the characteristic function through the tails entering
the Fourier integral.
To this end, we consider maturities at four different scales,
\[
    T \in \{0.019, 0.083,0.25,1\},
\]
roughly corresponding approximately to one week, one month, three months, and one year. 
As is well known, the Fourier pricing of short-maturity options requires a
larger contour cutoff $M$ to achieve reasonable accuracy. We therefore restrict
attention to a range of  maturity-dependent strikes.
Specifically, for each maturity $T$ we take $76$ strikes with uniformly spaced
log-moneyness
\[
    k:=\log(K/ S_0) \in \sqrt{T}\,[-1.0,0.5].
\]
In addition to plotting the implied-volatility smiles, we record for each pair
$(H,T)$ and each considered expansion order $N$ the largest relative implied-volatility
error over the strike grid. That is, for each strike we take the absolute
difference between the implied volatility obtained from the expansion and the
benchmark implied volatility, divide it by the benchmark implied volatility, and
then take the maximum over the $76$ strikes in the maturity-dependent
log-moneyness grid.

\textit{Convergence.}
To assess the behaviour of the $H$-expansion at the Riccati level, we use
coefficient ratios associated with the Taylor polynomial
$P^{H}_{N,H_0}$ in \eqref{eq:taylor}. We first estimate a convergence
radius along the Fourier contour $z=R-iu$. For fixed $u$, let
\[
    \bar f_n^{H_0}(u)
    :=
    \max_{0\leq t\leq T}
    \left|
        \partial_H^n\psi(t;H_0,R-iu)
    \right|,
\]
where the maximum is taken over the time grid used in the Riccati computation with $T=1$ and $J=500$ steps.
We estimate the radius of convergence by the coefficient ratios
\[
    R_{\mathrm{est}}^{H_0}(n,u)
    :=
    \frac{(n+1)\bar f_n^{H_0}(u)}
         {\bar f_{n+1}^{H_0}(u)} .
\]
For each $u$, we compute $R^{H_0}_{\mathrm{est}}(n,u)$ up to order $n=19$ and define
$R^*_{H_0}(u)$ as the mean of the first window of five consecutive estimates whose
range is smaller than $10^{-2}$; if no such window exists, no radius estimate is
assigned.

As a second diagnostic, we consider the tail ratio of the final Taylor terms,
\[
    \max_{n=N-3,\ldots,N-1}
    \left(
    \frac{|H- H_0|}{n+1}
    \frac{\bar f^{H_0}_{n+1}(u)}{\bar f^{H_0}_n(u)}
    \right).
\]
This quantity may serve as a a local indicator for the convergence speed of the
Taylor expansion at the Riccati level.

\subsection{Results}

We start by discussing the  convergence of the Riccati
expansion, and then present the results for pricing call options with the
expansion approach.

Figure~\ref{fig:ratio_heatmap} shows the tail-ratio diagnostic for the final
Taylor contributions of the order-$12$ expansion, together with the estimated
convergence boundary $H=H_0\pm R^*_{H_0}(|u|)$. As expected, the estimated convergence
region narrows as the Fourier frequency $|u|$ increases. Nevertheless, the
resulting intervals remain sufficiently wide even for large $u$, making the
expansion numerically meaningful over the Fourier range relevant for option pricing.
In fact, the observed convergence intervals are substantially wider than the
theoretical bounds proved in Theorem \ref{thm:errorbound}.

\begin{figure}[H]
    \centering
    \includegraphics[width=\linewidth]{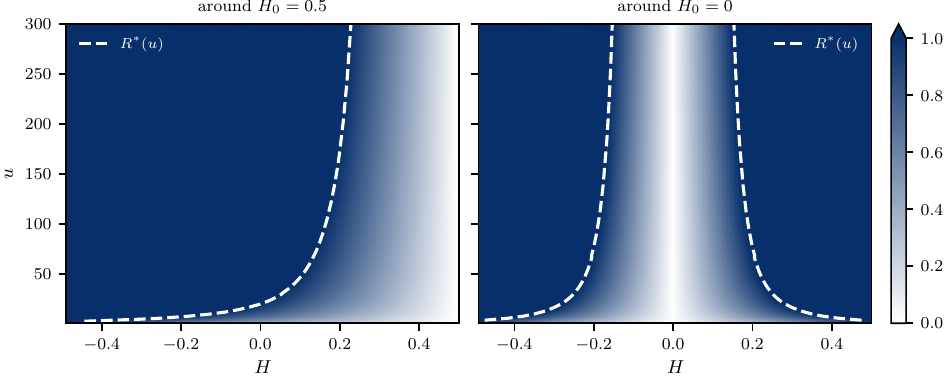}
    \caption{Riccati-level tail-ratio diagnostic for the $H$-expansion of order
    $N=12$ around the two expansion points $H_0=1/2$ and $H_0=0$. The
    colour scale shows the tail ratio of the final Taylor contributions, and
    the overlaid curves show the estimated convergence boundary
    $H=H_0\pm R^*(u)$. The heatmaps are computed on a grid of $100$ Hurst
    values and $300$ Fourier frequencies; the maximum over time is taken over
    the graded Riccati time grid on $[0,1]$ with $500$ time steps.}
    \label{fig:ratio_heatmap}
\end{figure}

The expansion approach makes it straightforward to vectorize the computation of
approximate Riccati solutions over several values of $H$. Indeed, for each
considered expansion point $H_0$, computing the full array of expanded Riccati
solutions used to create the heatmap in Figure~\ref{fig:ratio_heatmap}, at
expansion order $N=12$, over $100$ Hurst parameters, $300$ Fourier frequencies,
and $500$ time steps, took less than four seconds\footnote{On an Apple M3 CPU
with $8$ cores and $24$GB of memory, using a JAX implementation.}. This suggests
that the Riccati-level diagnostic provides a cheap heuristic for identifying the
region in which the expansion is expected to converge and can therefore be used
for pricing.

In Figures~\ref{fig:expansion-anchor0p5} and~\ref{fig:expansion-anchor0} we
plot implied volatilities computed with the approximation
\eqref{eq:call_approx} for the expansions around $H_0=1/2$ and $H_0=0$,
respectively, and compare them against benchmark implied volatilities computed
using a direct discretization of the Riccati equation; see
Section~\ref{sec:implementation}.  
For completeness, we also report the maximal relative errors in implied
volatilities for each pair of $H$ and $T$ in Table~\ref{tab:impld_vol_H0p5} and \ref{tab:impld_vol_H0p0}.

For the case $H_0=1/2$, we show implied
volatilities for $H\in\{0.4,0.3,0.2\}$.
We observe that the smiles obtained
with expansion orders $N=2$ and $N=4$ visually overlap with the benchmark for
$H=0.4$ across all considered maturities, and for maturity $T=1$ across all
considered values of $H$. Only for smaller $H$ and short maturities do we start
to observe a mismatch in the wings of the smile, first for expansion order
$N=2$ and then, for $(H,T)=(0.2,0.083)$ and $(H,T)=(0.3,0.019)$, also a visible
mismatch for expansion order $N=4$. At $(H,T)=(0.2,0.019)$, the $u$-range
required by the adaptive Fourier pricing extends beyond the estimated
convergence radius of the expansion.

\begin{figure}[H]
    \centering
    \includegraphics[width=\textwidth]{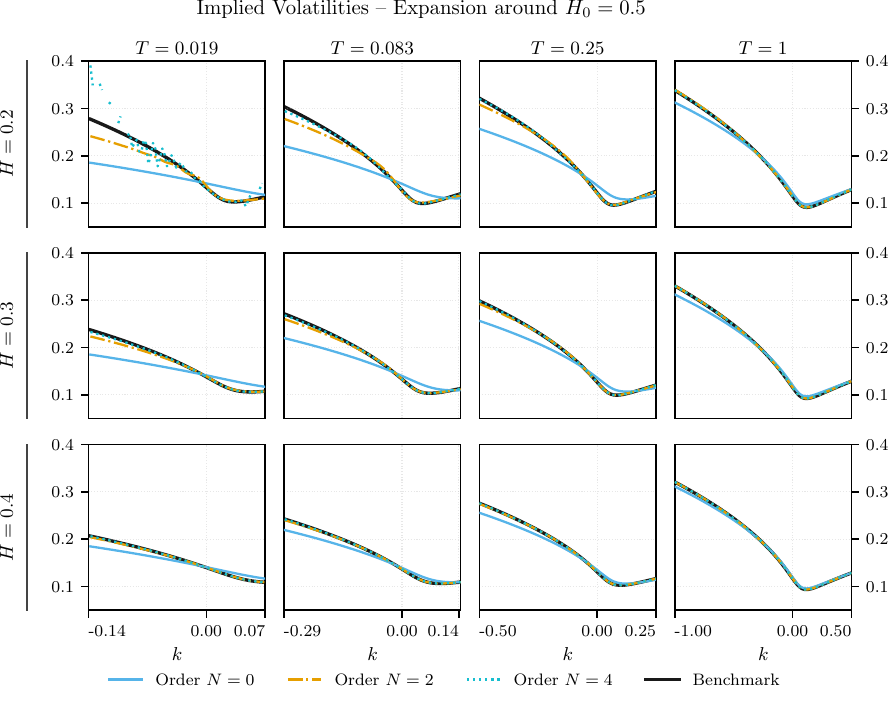}
    \caption{Implied-volatility smiles for the expansion around $H_0 H=1/2$.
    Rows correspond to the target Hurst parameters
    $H\in\{0.4,0.3,0.2\}$ and columns to the maturities
    $T\in\{0.019, 0.083,0.25,1\}$. The expansion orders
    $N\in\{0,2,4\}$ are compared with benchmark implied volatilities computed
    using a direct discretization of the Riccati equation.}
    \label{fig:expansion-anchor0p5}
\end{figure}

\begin{table}[H]
\centering
\small
\caption{Maximal relative implied volatility errors, expansion around $H_0 = 0.5$.}\label{tab:impld_vol_H0p5}
\label{tab:implvols_0p5}
\begin{tabular}{lrr@{\quad}rr@{\quad}rr@{\quad}rr}
\toprule
$H$ & \multicolumn{2}{c}{$T = 0.019$} & \multicolumn{2}{c}{$T = 0.083$} & \multicolumn{2}{c}{$T = 0.25$} & \multicolumn{2}{c}{$T = 1$} \\
\cmidrule(lr){2-3}  \cmidrule(lr){4-5}  \cmidrule(lr){6-7}  \cmidrule(lr){8-9}
 & $N=2$ & $N=4$ & $N=2$ & $N=4$ & $N=2$ & $N=4$ & $N=2$ & $N=4$ \\
\midrule
$0.2$ & $0.1295$ & $0.3965$ & $0.0841$ & $0.0293$ & $0.0415$ & $0.0070$ & $0.0057$ & $0.0010$ \\
$0.3$ & $0.0586$ & $0.0187$ & $0.0408$ & $0.0105$ & $0.0209$ & $0.0019$ & $0.0020$ & $0.0002$ \\
$0.4$ & $0.0120$ & $0.0017$ & $0.0092$ & $0.0010$ & $0.0048$ & $0.0002$ & $0.0003$ & $0.0000$ \\
\bottomrule
\end{tabular}
\end{table}

A similar picture is seen for $H_0=0$, where we consider implied volatilities
for Hurst parameters $H\in\{-0.3,-0.2,-0.1,0.1,0.2\}$. The overall behaviour is
similar: the volatility smiles match the benchmark  closely for larger
maturities and for values of $H$ closer to zero. Interestingly, the expansion
appears to be more accurate in the negative-$H$ range. For $H=-0.1$, the smiles
for expansion orders $N=2$ and $N=4$ visually match the benchmark across all
maturities, while for $H=-0.2$ a slight mismatch becomes visible only at
$T=0.019$.
This is promising for the expansion approach, since direct Riccati
discretizations become particularly delicate in the negative-$H$ regime.

\begin{figure}[H]
    \centering
    \includegraphics[width=\textwidth]{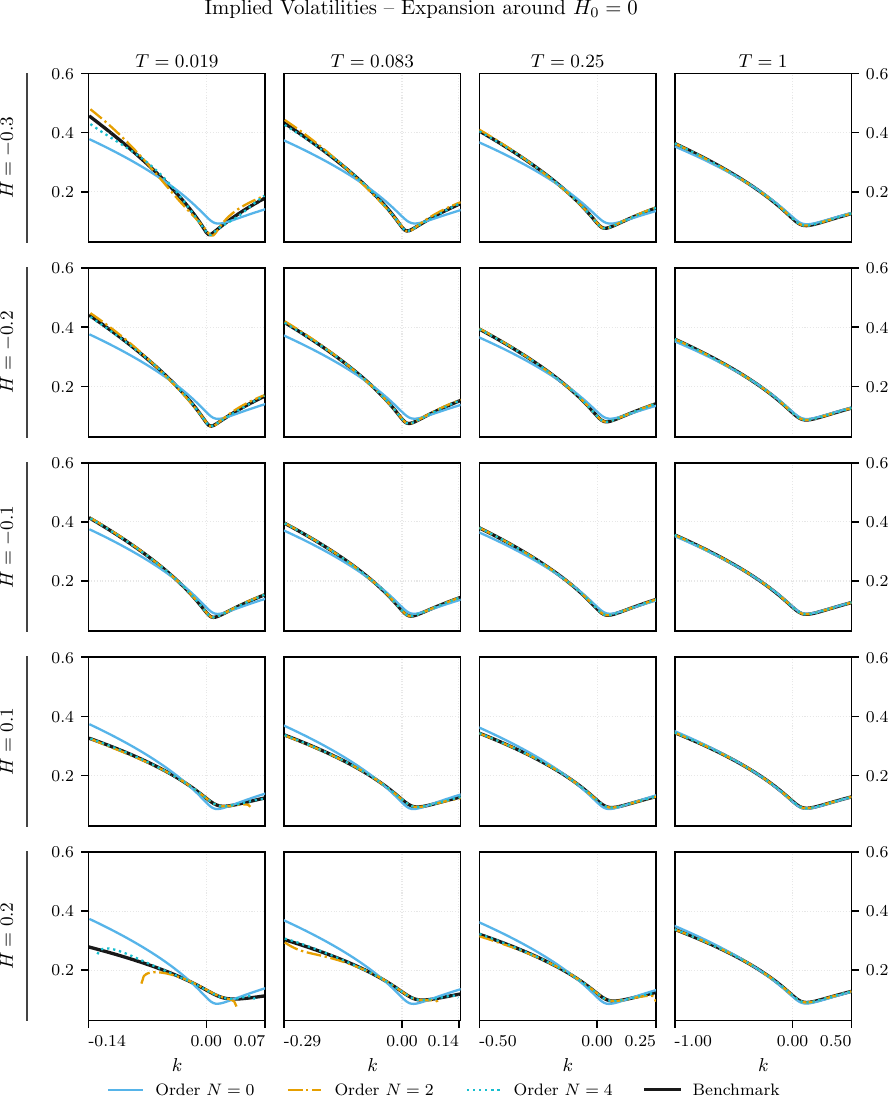}
    \caption{Implied-volatility smiles for the expansion around $H_0=0$.
    Rows correspond to the target Hurst parameters
    $H\in\{-0.3,-0.2, -0.1, 0.1,0.2\}$ and columns to the maturities
    $T\in\{0.019, 0.083,0.25,1.0\}$. The expansion orders
    $N\in\{0,2,4\}$ are compared with benchmark implied volatilities computed
    using a direct discretization of the Riccati equation.}
    \label{fig:expansion-anchor0}
\end{figure}

For $H=0.2$ we observe an interesting boundary behaviour. The $u$-range required
for the maturities $T=0.083$ and $T=0.019$ appears to remain within, but close
to, the estimated convergence radius, and the approximation improves with
increasing expansion order. We verified numerically that the smiles eventually
match the benchmark at orders $N=8$ and $N=10$, respectively. By contrast, for
the same maturities and $H=-0.2$, the expansion converges much faster.
This asymmetry in effective convergence speed is likely explained by the alternating signs in the expansion for $H<0$. For
$H=-0.3$, the $u$-range required for $T=0.083$ and $T=0.019$ reaches outside the
estimated convergence radius; nevertheless, the implied volatilities for the
displayed expansion orders still remain fairly close to the benchmark.

\begin{table}[H]
\centering
\small
\caption{Maximal relative implied volatility errors, expansion around $H_0 = 0$.}
\label{tab:impld_vol_H0p0}
\label{tab:implvols_0}
\begin{tabular}{lrr@{\quad}rr@{\quad}rr@{\quad}rr}
\toprule
$H$ & \multicolumn{2}{c}{$T = 0.019$} & \multicolumn{2}{c}{$T = 0.083$} & \multicolumn{2}{c}{$T = 0.25$} & \multicolumn{2}{c}{$T = 1$} \\
\cmidrule(lr){2-3}  \cmidrule(lr){4-5}  \cmidrule(lr){6-7}  \cmidrule(lr){8-9}
 & $N=2$ & $N=4$ & $N=2$ & $N=4$ & $N=2$ & $N=4$ & $N=2$ & $N=4$ \\
\midrule
$-0.3$ & $0.1816$ & $0.0674$ & $0.0512$ & $0.0138$ & $0.0155$ & $0.0065$ & $0.0030$ & $0.0027$ \\
$-0.2$ & $0.0562$ & $0.0150$ & $0.0212$ & $0.0073$ & $0.0075$ & $0.0054$ & $0.0025$ & $0.0025$ \\
$-0.1$ & $0.0114$ & $0.0074$ & $0.0074$ & $0.0064$ & $0.0042$ & $0.0040$ & $0.0025$ & $0.0025$ \\
$0.1$ & $0.2040$ & $0.0257$ & $0.0337$ & $0.0036$ & $0.0038$ & $0.0031$ & $0.0025$ & $0.0025$ \\
$0.2$ & $0.3180$ & $0.0611$ & $0.2203$ & $0.1697$ & $0.0954$ & $0.0032$ & $0.0024$ & $0.0026$ \\
\bottomrule
\end{tabular}
\end{table}

\section{Summary and outlook}

In this paper, we have proven that the solution to the fractional Riccati equation and hence also the characteristic function of the log stock price in the rough Heston model are real analytic in the Hurst parameter $H>-1/2$. We have then expanded the solution of the fractional Riccati equation, using that the $n$-th partial derivative with respect to $H$ solves a linear Volterra equation. Our theoretical error estimate suggests that the radius of convergence of the Taylor series should be decreasing in the size of the Fourier coefficients. This is confirmed by our numerical experiments, according to which the convergence radius is larger than $0.25$ (resp.~$0.15$) for $|u|\leq 300$ if we expand in $H=0.5$ (resp.~$H=0$). 
We have then replaced the solution to the fractional Riccati equation by its $N$-th Taylor polynomial in the Fourier pricing formula and computed implied volatilities. The resulting numerical scheme is very fast and easy to implement and our numerical results suggest that the approximation works very well for already small values of $N$ like $N=2,3,4$.

While our numerical results are very strong, it is not possible to infer the exact convergence radius from our theoretical error estimate. The reason is that we rely on recursive estimates to control the size of the $n$-th derivative of the Riccati solution, which due to the quadratic structure  naturally results in large bounds. For this reason, we did also not provide weak error bounds for option prices, which could of course be derived from the error bound for the approximation of the fractional Riccati solution using similar techniques as in \cite{bayer2023rough}. We leave both questions for future research.
Finally, it may also be interesting to apply our methodology to the pricing of path-dependent options in the rough Heston model.

\appendix
\section{Quadratic recurrence relation}

The estimate of the error term of the Taylor expansion relies on the following result. 

\begin{Lem}\label{lem:recurrence}
 If for given  $\alpha_1,\alpha_2,c_1>0$ the sequence $c=(c_n)_{n\in\N}$ satisfies the quadratic recurrence
 \begin{align*}
     c_n=\alpha_1c_{n-1}+\alpha_2\sum_{k=1}^{n-1}c_kc_{n-k},\quad n\geq2,
 \end{align*}
 then for $n\geq2$,
\begin{align*}
     c_n&=\frac{(-1)^n\alpha_1^n}{\alpha_2(2+4\alpha_2c_0)^n}\sum_{k=\lceil n/2\rceil}^n\binom{k}{n-k}\binom{2k-1}{k}\frac{(-1)^{k}}{2k-1}(1+2\alpha_2c_0)^{2k}\\
     &=\begin{cases}
     \frac{(-1)^{n/2}}{\alpha_2}\left(\frac{\alpha_1}{2}\right)^n\binom{n-1}{n/2}\frac{1}{n-1} {}_2F_1\left(\frac{n-1}{2},-\frac{n}{2},\frac{1}{2},(1+\frac{2\alpha_2c_1}{\alpha_1})^2\right)&:n\in 2\N,\\
     \frac{(-1)^{(n-1)/2}}{\alpha_2}(1+\frac{2\alpha_2c_1}{\alpha_1})\left(\frac{\alpha_1}{2}\right)^n\binom{n}{(n+1)/2}\frac{n+1}{2n} {}_2F_1\left(\frac{1-n}{2},\frac{n}{2},\frac{3}{2},(1+\frac{2\alpha_2c_1}{\alpha_1})^2\right)&:n\in 2\N-1,
     \end{cases}
     \end{align*}
 where $_2F_1$ denotes the ordinary hypergeometric function. Moreover, 
 \begin{equation*}
 0<c_n\leq \frac{2^n(\alpha_1+2\alpha_2c_1)^n}{2(n-1)}\cdot\left(1+\frac{1}{2(1+\frac{2\alpha_2c_1}{\alpha_1})^2}\right)^{\lfloor n/2\rfloor}.
 \end{equation*}
 i.e.~$c_n$ is of exponential growth as $n\rightarrow\infty$
\end{Lem}

\begin{proof}
Let us define a new sequence $d=(d_n)_{n\in\N}$ by 
\begin{align*}
d_1&:=\frac{\alpha_1}{2}+\alpha_2c_1,\\
d_n&:=\alpha_2 c_n,\quad n\geq2.
\end{align*}
Then $d$ satisfies the quadratic recurrence
\[d_n=\sum_{i=1}^{n-1}d_id_{n-i},\quad n\geq 3.\]
Hence, the generating function of this recurrence
\[d(x)=\sum_{n=1}^\infty d_nx^n\]
satisfies on its domain of definition the equation
\[d(x)=d^2(x)+d_1x-\frac{\alpha_1^2}{4}x^2,\]
which possesses the two solutions
\[d_{1/2}(x)=\frac{1}{2}\left(1\pm\sqrt{1+\alpha_1^2x^2-4d_1x}\right).\]
Due to $d(0)=0$, we must have
\begin{align*}
    d(x)&=\frac{1}{2}\left(1-\sqrt{1+\alpha_1^2x^2-4d_1x}\right)\\
    &=-\frac{1}{2}\sum_{k=1}^\infty\binom{1/2}{k}(\alpha_1^2x^2-4d_1x)^k\\
    &=-\frac{1}{2}\sum_{k=1}^\infty\binom{1/2}{k}x^k\sum_{j=0}^k\binom{k}{j}(\alpha_1^2x)^j(-4d_1)^{k-j}\\
    &=-\frac{1}{2}\sum_{n=1}^\infty\frac{\alpha_1^{2n}(-1)^nx^n}{4^nd_1^n}\sum_{k=\lceil n/2\rceil}^n\binom{1/2}{k}\binom{k}{n-k}\left(-\frac{4d_1}{\alpha_1}\right)^{2k}\\
    &=\sum_{n=1}^\infty\frac{(-1)^n\alpha_1^nx^n}{(4d_1/\alpha_1)^n}\sum_{k=\lceil n/2\rceil}^n\binom{k}{n-k}\binom{2k-1}{k}\frac{(-1)^k}{2k-1}\left(\frac{2d_1}{\alpha_1}\right)^{2k}.
\end{align*}
Matching coefficients we obtain
\begin{align*}
     d_n%
     &=\frac{(-1)^n\alpha_1^n}{(2+\frac{4\alpha_2c_1}{\alpha_1})^n}\sum_{k=\lceil n/2\rceil}^n\binom{k}{n-k}\binom{2k-1}{k}\frac{(-1)^{k}}{2k-1}\left(1+\frac{2\alpha_2c_1}{\alpha_1}\right)^{2k}.
 \end{align*}
 To derive the asymptotic behaviour, we note that for $n\geq 2$,
 \begin{align*}
     d_n&=\frac{(-1)^n\alpha_1^n}{(2+\frac{4\alpha_2c_1}{\alpha_1})^n}\sum_{k=\lceil n/2\rceil}^n\binom{k}{n-k}\binom{2k-1}{k}\frac{(-1)^{k}}{2k-1}\left(1+\frac{2\alpha_2c_1}{\alpha_1}\right)^{2k}\\
     &\leq \frac{\alpha_1^n}{2(2+\frac{4\alpha_2c_1}{\alpha_1})^n}\sum_{j=0}^{\lfloor n/2\rfloor}\binom{n-j}{j}\frac{4^{n-j}(1+\frac{2\alpha_2c_1}{\alpha_1})^{2(n-j)}}{2n-2j-1}\\
     &\leq\frac{2^n\alpha_1^n(1+\frac{2\alpha_2c_1}{\alpha_1})^n}{2(n-1)}\sum_{j=0}^{\lfloor n/2\rfloor}\binom{n-j}{j}4^{-j}\left(1+\frac{2\alpha_2c_1}{\alpha_1}\right)^{-2j}\\
     &\leq \frac{2^n(\alpha_1+2\alpha_2c_1)^n}{2(n-1)}\sum_{j=0}^{\lfloor n/2\rfloor}\binom{\lfloor n/2\rfloor}{j}2^{-j}\left(1+\frac{2\alpha_2c_1}{\alpha_1}\right)^{-2j}\\
     &=\frac{2^n(\alpha_1+2\alpha_2c_1)^n}{2(n-1)}\cdot\left(1+\frac{1}{2(1+\frac{2\alpha_2c_1}{\alpha_1})^2}\right)^{\lfloor n/2\rfloor}.
 \end{align*}
\end{proof}

\section{Fractional Volterra equations}

\subsection{Bounds for complex-valued, linear Volterra equations}

For the proof of Theorem \ref{thm:errorbound} it is important to have a priori bounds for the solution of the parameter-dependent fractional Riccati equation $\psi$. The following lemma follows from \cite[Theorem C.3, Corollary C.4]{AJEE19}, where the result is shown for constant source terms. 

\begin{Lem}\label{lem:linearvolterrabound} Let $k\in L^2([0,T],\R)$ be non-negative, non-increasing, and
continuous on $(0,T]$ and suppose that its resolvent of the first kind  is non-negative and non-increasing. Let $h_0,z:[0,T]\to\C$ be continuous functions such that $\text{Re}(z)\leq\lambda$ for some $\lambda\in\R$. Then the equation 
    \[h(t)=h_0(t)+\int_0^tk(t-s)z(s)h(s)ds,\quad t\in[0,T],\]
    has a unique continuous solution $h:[0,T]\to\C$, which satisfies for all $t\in[0,T]$,
    \[|h(t)|\leq |h_0(t)|+\|h_0z\|_{\infty,T}\int_0^tE_\lambda(s)ds,\]
    where $E_\lambda$ denotes the canonical resolvent of $k$ with respect to $\lambda$.
\end{Lem}

\begin{proof}
According to \cite[Corollary C.4]{AJEE19}, the equation 
    \[\bar h(t)=\int_0^tk(t-s)\left(z(s)\bar h(s)+z(s)h_0(s)\right)ds,\quad t\geq0,\]
    has a unique continuous solution $\bar h:\R_+\to\C$ which satisfies for all $t\in[0,T]$,
    \[|\bar h(t)|\leq \|h_0z\|_{\infty,t}\int_0^tE_\lambda(s)ds.\]
Since $h:=\bar h+h_0$ solves the equation 
    \[h(t)=h_0(t)+\int_0^tk(t-s)z(s)h(s)ds,\quad t\in[0,T],\]
    it satisfies for all $t\in[0,T]$,
    \[|h(t)|\leq |h_0(t)|+\|h_0z\|_{\infty,t}\int_0^tE_\lambda(s)ds.\]
\end{proof}

\subsection{Differentiability of parameter-dependent Volterra equations}

Let us recall from \cite[Section 9]{GLS90} that a Volterra kernel of continuous type on $[0,T]$ is any measurable function $k:[0,T]^2\to\R^{n\times n}$ such that $k(t,s)=0$ for all $t>s$ and such that for all $t\in[0,T]$,
\[\int_0^T|k(t,s)|ds<\infty\quad\text{and}\quad \int_0^T|k(t+h,s)-k(t,s)|ds\stackrel{h\to0}{\longrightarrow}0.\]

In the following we consider the parameter-dependent Volterra integral equation 
\begin{equation}\label{eq:pVol}
x(\lambda,t)=a(\lambda,t)+\int_0^tk(\lambda,t,s)f(\lambda,s,x(\lambda,s))ds,\quad t\in[0,T],\quad \lambda\in \Lambda,
\end{equation}
where $\Lambda\subset\R$ is an open interval, $a:\Lambda\times[0,T]\to\R^n$ and $f:\Lambda\times[0,T]\times\R^n\to\R^n$. 

\begin{The}[cf.~Theorems 13.1.1 and 13.1.2 in \cite{GLS90}]\label{thm:Voldif}
Suppose that for each $\lambda\in \Lambda$, $k(\lambda,\cdot)$ is a Volterra kernel of continuous type.  
\begin{enumerate}
\item[a)] Assume that $a\in C(\Lambda\times[0,T],\R^n)$ and $f\in C(\Lambda\times[0,T]\times\R^n,\R^n)$. If for every $\lambda\in\Lambda$,
\[\sup_{t\in[0,T]}\int_0^t|k(\lambda+\varepsilon,t,s)-k(\lambda,t,s)|ds\stackrel{\varepsilon\to0}{\longrightarrow}0\]  
and equation \eqref{eq:pVol} has a unique solution $x(\lambda,\cdot)$ on $[0,T]$, then $x\in C(\Lambda\times[0,T],\R^n)$, i.e.~$x$ depends continuously on $\lambda$ and $t$.
\item[b)] Assume that $a\in C^{1,0}(\Lambda\times[0,T],\R^n)$ and $f\in C^{1,0,1}(\Lambda\times[0,T]\times\R^n,\R^n)$. If there exists $k_\lambda:\Lambda\times[0,T]^2\to \R^{n\times n}$ such that for each $\lambda\in\Lambda$, $k_\lambda(\lambda,\cdot)$ is a Volterra kernel of continuous type on $[0,T]$ satisfying as $\varepsilon\to0$,
\[\sup_{t\in[0,T]}\int_0^t|k(\lambda+\varepsilon,t,s)-k(\lambda,t,s)-\varepsilon k_\lambda(\lambda,t,s)|ds=o(\varepsilon),\]
then for each $\lambda\in\Lambda$ there exists a unique solution $x(\lambda,\cdot)$ of \eqref{eq:pVol} on $[0,T]$. Moreover, $x\in C^{1,0}(\Lambda\times[0,T],\R^n)$, i.e.~$x$ is continuously differentiable with respect to $\lambda$, and the derivative satisfies
\begin{equation}
\begin{split}
x_\lambda(\lambda,t)&=a_\lambda(\lambda,t)+\int_0^t\left(k_\lambda(\lambda,t,s)f(\lambda, s,x(\lambda,s))+k(\lambda,t,s)f_\lambda(\lambda,s,x(s,\lambda))\right)ds\\
&\qquad+\int_0^tk(\lambda,t,s)f_x(\lambda,s,x(\lambda,s))x_\lambda(\lambda,s)ds,\quad t\in[0,T].
\end{split}
\end{equation}
\end{enumerate}
\end{The}

\section{A priori bound for the (rough) Heston Riccati equation}

\subsection{Bound for the fractional Riccati equation \eqref{eq:riccati}}

We recall that the canonical resolvent of a kernel $K\in L^1_{loc}(\R_+,\R)$ with parameter $\mu\in\R$ is defined as the unique solution $E_\mu\in L^1_{loc}(\R_+,\R)$ of the convolution equation 
\[E_\mu-K=\mu K\ast E_\mu,\]
cf.~\cite[Theorem 2.3.1]{GLS90}. Especially, the canonical resolvent $E_\mu$ of $K^H$ for $H\in(0,1/2)$ is given by 
\[E^H_\mu(t)=t^{H-1/2}\sum_{k=0}^\infty\frac{(\mu t^{H+1/2})^k}{\Gamma((k+1)(H+1/2))}.\]
The following lemma is an adaptation of \cite[Proposition 5.4]{AJEE19} to the solution of the parameter-dependent fractional Riccati equation \eqref{eq:riccati}, spelling out the dependence on $H$ explicitly.

\begin{Lem}\label{lem:boundpsi}
The solution $\psi$ to \eqref{eq:riccati} satisfies $\Re(\psi(\cdot;H,z))\leq0$ for all $H\in(-1/2,\infty)$ and there exists for all $H_0\in(-1/2,\infty)$ and $T\in[0,1]$ a constant $C_{T,H_0}>0$ (not depending on $z$) such that
\[\sup_{(t,H)\in[0,T]\times[H_0,\infty)}\left|\psi\left(t; H,z\right)\right|\leq C_{T,H_0}\left(\frac{1}{2}+b^2\right).\]
Moreover, $C_{T,H}$ depends continuously on $T,H$ and satisfies $C_{T,H}\to0$ as $T\to0$.
\end{Lem}

\begin{proof}
Working along the lines of the proof of \cite[Proposition 5.4]{AJEE19} we obtain for any $H\in(-1/2,\infty)$ that $\Re(\psi(\cdot;H,z))\leq 0$ and 
\[\sup_{t\in[0,T]}\left|\psi(t;H,z)\right|\leq \frac{|z^2-z|}{2}\int_0^TE^{H}_{\mu}(s)ds,\]
where $E^H_{\mu}$ denotes the canonical resolvent of $K^H$ with parameter $\mu:=\nu-\lambda\leq0$. We now estimate
\begin{align*}
\int_0^TE^H_\mu(s)ds&= 
\sum_{k=0}^\infty\frac{\mu^kT^{(k+1)(H+1/2)}}{\Gamma((k+1)(H+1/2)+1)}\\
&\leq T^{H+1/2}\sum_{k=0}^\infty\frac{\mu^k}{\Gamma((k+1)(H+1/2)+1)}
\leq T^{H+1/2}C(H)<\infty,
\end{align*}
where $C$ is increasing in $H$ with $C(H)\to\infty$ as $H\to-1/2$. 
Hence, for any $H_0\in(-1/2,\infty)$ and $T\in[0,1]$ there exists a constant $C_{T,H_0}>0$ satisfying $C_T\rightarrow0$ as $T\rightarrow0$ such that 
\[\sup_{(t,H)\in[0,T]\times[H_0,\infty)}\left|\psi(t;H,z)\right|\leq C_{T,H_0}|z^2-z|.\]
Finally, we note that for $z=a+ib\in\C^*$, one has
\[|z^2-z|\leq \frac{1}{2}+b^2.\]
\end{proof}

\subsection{Bound for the standard Heston Riccati equation}

While  the a priori estimate for the fractional Riccati equation (Lemma \ref{lem:boundpsi}) gives a quadratic dependence on the Fourier argument $|b|$, the statement can be improved in the case $H=1/2$ to a linear dependence on $|b|$. The following result should be well-known, but we could not find the precise statement in the literature. So we provide a proof. 

\begin{Lem}\label{lem:boundHeston}
   Let $\psi(\cdot;1/2,z)$ denote the unique solution to the Riccati equation
    \[\phi(t;z)=\int_0^t{F}\left(z,\phi(s;z)\right)ds,\qquad t\geq0,\qquad z=a+ib\in\C^*.\]
    Then there exists a constant $C>0$ such that for every $t\geq0$ and $z=a+ib\in\C^*$,
    \[|\psi(t;1/2,z)|\leq C(1+|b|).\]
\end{Lem}

\begin{proof}
Recall from \cite{H93} that the solution $\phi$ is explicitly given by 
    \begin{align*}
        \psi(t;1/2,z)=\frac{\beta(z)-d(z)}{\nu^2}\cdot\frac{1-e^{-d(z)t}}{1-g(z)e^{-d(z)t}}
    \end{align*}
    with
    \[d(z)=\sqrt{\beta^2(z)-\nu^2(z^2-z)},\qquad \beta(z):=\lambda-\rho\nu z,\qquad g(z):=\frac{\beta(z)-d(z)}{\beta(z)+d(z)}.\]
    where $d(z)$ is chosen to be the principal branch of the square root, i.e.~$\Re(d(z))\geq0$, cf.~\cite[Theorem 3]{KL06}. Hence, 
        \[\left|1-e^{-d(z)t}\right|\leq 1+|e^{-\Re{d(z)}t}|\leq 2.\]
Clearly, $\beta(z)$ grows linearly in $|b|$. Moreover, we have
    \begin{align*}
       \Re(d^2(z))&=\nu^2(1-\rho^2)b^2+(\lambda-\rho\nu a)^2+\nu^2(a-a^2)\geq (\lambda-\rho\nu)^2>0,\\ 
       \Im(d^2(z))&=\left(\nu^2(1-2a)-2\rho\nu(\lambda-\rho\nu a)\right)b,
    \end{align*}
    and by \cite[Lemma A.1]{KL06},
    \[\Re(d(z))=\sqrt{\frac{1}{2}\Re(d^2(z))+\frac{1}{2}\sqrt{(\Re(d^2(z))^2+(\Im(d^2(z))^2}},\qquad \Im(d(z))=\frac{\Im(d^2(z))}{2\Re(d(z))}.\]
    Hence, as $|b|\to\infty$,
    \[\Re(d^2(z))\sim\nu^2(1-\rho^2)b^2,\qquad |\Re(d(z))|\sim \nu\sqrt{1-\rho^2}|b|,\qquad |\Im(d(z))|\to\frac{|\nu(1-2a)-2\rho(\lambda-\rho\nu a)|}{2\sqrt{1-\rho^2}}.\]
As  $\Re(d(z))\geq\sqrt{\Re(d^2(z))}\geq \lambda-\rho\nu>0$, $\Im(d(z))$ is indeed bounded for $z\in\C^*$. This shows that $d(z)$ grows linearly in $|b|$ as well. Thus, it remains to show that the term in the denominator, i.e.~$1-g(z)e^{-d(z)t}$, is bounded away from zero, uniformly in $z\in\C^*$ and $t\geq0$. 

As $\Im(d(z))$ is bounded for $z\in\C^*$, we may choose for any $\epsilon>0$ a small $t_0=t_0(\epsilon)>0$ such that $|\sin(\Im(d(z))t)|<\epsilon$ and $\cos(\Im(d(z)t)\geq0$ for all $t< t_0$. On the other hand, we have for all $t\geq t_0$ that
    \[\left|1-g(z)e^{-d(z)t}\right|\geq 1-|g(z)|e^{-\Re(d(z))t_0}\geq 1-e^{-(\lambda -\rho\nu)t_0}>0.\]
It remains to show the bound for $t< t_0$. To this end, we differentiate two cases: $\rho\neq0$ and $\rho=0$. 

First suppose that $|\rho|\in(0,1)$: It is shown in \cite[Lemma 3]{KL06} that $|g(z)|\leq 1$ for all $z\in\C$. 
       Therefore, 
        \begin{align*}
            \left|1-g(z)e^{-d(z)t}\right|&=\left|1-g(z)e^{-\Re(d(z))t}(\cos(\Im(d(z))t)-i\sin(\Im(d(z))t))\right|\\
            &\geq \left|1-g(z)e^{-\Re(d(z))t}\cos(\Im(d(z))t)\right|-\left|g(z)e^{-\Re(d(z))t}\sin(\Im(d(z))t)\right|\\
            &\geq \left|\Re\left(1-g(z)e^{-\Re(d(z))t}\cos(\Im(d(z))t)\right)\right|-\left|\sin(\Im(d(z))t)\right|\\
            &\geq 1-\left|\Re\left(g(z)\right)\right|-\epsilon
        \end{align*}
        for all $t<t_0$. So, it suffices to show that $1-|\Re(g(z))|$ is bounded away from zero. We have
    \[\Re(g(z))=\frac{|\beta(z)|^2-|d(z)|^2}{|\beta(z)+d(z)|^2}\]
    and hence
    \begin{align*}
        1-|\Re(g(z))|= \frac{2\left(|\beta(z)|^2\wedge|d(z)|^2+\Re\left(\beta(z)\bar d(z)\right)\right)}{|\beta(z)+d(z)|^2}.
    \end{align*}
Since    
    \begin{align*}
1\geq |g(z)|^2=\frac{|\beta(z)-d(z)|^2}{|\beta(z)+d(z)|^2}=\frac{|d(z)|^2+{|\beta(z)|^2-2\Re\left(\beta(z)\bar d(z)\right)}}{|d(z)|^2+{|\beta(z)|^2+2\Re\left(\beta(z)\bar d(z)\right)}},  
\end{align*}
we have $\Re\left(\beta(z)\bar d(z)\right)>0$ and hence
    \begin{align*}
        1-|\Re(g(z))|\geq \frac{2\left(|\beta(z)|^2\wedge|d(z)|^2\right)}{|\beta(z)+d(z)|^2}.
    \end{align*}
First note that as $\Re(d(z))\geq \lambda-\rho\nu>0$ and $\Re(\beta(z))\geq \lambda-\rho\nu>0$, we have $1-|\Re(g(z))|>0$ for all $z\in\C^*$. Moreover, as $|b|\to\infty$, 
\begin{align*}
\frac{|\beta(z)|^2\wedge|d(z)|^2}{|\beta(z)+d(z)|^2}\sim\frac{(\Im(\beta(z)))^2\wedge(\Re(d(z)))^2}{(\Re(d(z)))^2+(\Im(\beta(z)))^2}\sim \frac{(\rho^2\wedge (1-\rho^2))\nu^2 b^2}{\nu^2 b^2}\to\rho^2\wedge(1-\rho^2)>0.
\end{align*}
Hence, by continuity we may conclude that there is $c_0>0$ such that $1-|\Re(g(z))|>c_0>0$ for all $z\in\C^*$. Choosing $\epsilon=c_0/2$ proves the statement for $\rho\neq0$.

If $\rho=0$, we have $\beta(z)=\lambda$ and hence, $g(z)\to-1$ as $|b|\to\infty$. We estimate for $t<t_0$,
\begin{align*}
 \left|1-g(z)e^{-d(z)t}\right|&\geq \left|1+e^{-d(z)t}\right|-\left|(1+g(z))e^{-d(z)t}\right| \\
 &\geq \Re(1+e^{-d(z)t})-|1+g(z)|\\
 &\geq 1+e^{-\Re(d(z))t_0}\cos(\Im(d(z))t)-\frac{2\lambda}{|\lambda+d(z)|}\geq 1-\frac{2\lambda}{\lambda+\Re(d(z))}\\
 &\geq 1-\frac{2\lambda}{\lambda+\sqrt{\lambda^2+\nu^2b^2}}.
\end{align*}
Hence, there is $b_0>0$ such that  $\left|1-g(z)e^{-d(z)t}\right|>c_0>0$ for all $z=a+ib\in\C^*$ with $|b|\geq b_0$. On the other hand, $\rho=0$ implies that 
\begin{align*}
    \lambda^2+\nu^2b^2\leq \Re(d^2(z))\leq\nu^2b^2+\lambda^2+\frac{\nu^2}{4},\qquad
    |\Im(d^2(z))|\leq\nu^2|b|,\\
    \lambda\leq \sqrt{\lambda^2+\nu^2b^2}\leq \Re(d(z))\leq\nu\left(|b|+\frac{1}{2}\right)+\lambda,\qquad
    |\Im(d(z))|\leq\nu,
\end{align*}
and hence
\begin{align*}
    |g(z)|^2=\frac{(\Im(d(z)))^2+(\lambda-\Re(d(z)))^2}{(\Im(d(z)))^2+(\lambda+\Re(d(z)))^2}
    \end{align*}
is bounded away from $1$ if $|b|\leq b_0$. This shows that also in this case
\begin{align*}
     \left|1-g(z)e^{-d(z)t}\right|\geq 1-|g(z)|>0,
\end{align*}
uniformly in $z=a+ib\in\C^*$ with $|b|\leq b_0$.

\end{proof}

\section{Numerical approximations of weakly singular Volterra equations}

\subsection{Linear Volterra equations with fractional-logarithmic kernels}
\label{sec:fraclogvolterra}

We briefly describe the numerical scheme used to solve the linear Volterra
equations that appear in the expansion of the fractional Riccati equation.
Starting from an approximate reference solution
\(\psi(\cdot;H_0,z)=\psi^{H_0}(\cdot,z)\) on a time grid
\(
  0=t_0<t_1<\cdots<t_J=T,
\)
we compute all expansion terms on the same grid.
Recalling Theorem~\ref{thm:ndif}, the kernels encountered in the recursion are
of the form
\[
  K(t)=t^{H_0-\frac12}\log^n(t),
  \qquad n\in\mathbb N_0,
\]
and finite linear combinations thereof.
The linear Volterra equations that appear in the recursion then have the generic
form
\begin{equation}\label{eq:lin_volterra_apx}
  f(t)
  =
  r(t)
  +
  \int_0^t K(t-s)g(s;z)f(s)\dd{s},
\end{equation}
where \(f\) denotes the unknown expansion coefficient at the current order,
\(r\) is the known inhomogeneity constructed from lower-order terms, and
\[
  g(s;z)
  :=
  F_x\left(z,\psi(s;H_0,z)\right)
  =
  \rho\nu z-\lambda
  +
  \nu^2\psi(s;H_0,z)
\]
is the common coefficient obtained by linearizing the Riccati nonlinearity
around the reference solution. We assume that approximations
\(g_j \approx g(t_j;z)\) and \(r_j \approx r(t_j)\) are available on the grid
from the previous expansion terms and the approximate reference solution. We now
derive an approximation scheme for \(f_j\approx f(t_j)\).

We discretize the Volterra operator by product integration. On each interval
\([t_m,t_{m+1}]\), the integrand is approximated by linear interpolation. For
\(m<j\), define
\[
  \omega^L_{j,m}
  :=
  \int_{t_m}^{t_{m+1}}
  K(t_j-s)\frac{t_{m+1}-s}{t_{m+1}-t_m}\dd{s},
  \qquad
  \omega^R_{j,m}
  :=
  \int_{t_m}^{t_{m+1}}
  K(t_j-s)\frac{s-t_m}{t_{m+1}-t_m}\dd{s},
\]
and the lower-triangular product-integration matrix \(\Omega\) by
\[
  \Omega_{j,\ell}
  =
  \sum_{m=0}^{j-1}
  \left(
    \mathbf 1_{\{\ell=m\}}\omega^L_{j,m}
    +
    \mathbf 1_{\{\ell=m+1\}}\omega^R_{j,m}
  \right),
  \qquad 0\le \ell\le j,
\]
and set \(\Omega_{j,\ell}=0\) for \(\ell>j\).

The product-integration approximation of \eqref{eq:lin_volterra_apx} is
obtained by applying this matrix to the grid function obtained from linear
interpolation of \(g_\ell f_\ell\) on the grid, i.e.,
\[
  f_j
  =
  r_j
  +
  \sum_{\ell=0}^j \Omega_{j,\ell}g_\ell f_\ell.
\]
Solving for the unknown \(f_j\), we thus obtain the forward substitution
\[
  f_j
  =
  \frac{
    r_j+\sum_{\ell=0}^{j-1}\Omega_{j,\ell}g_\ell f_\ell
  }{
    1-\Omega_{j,j}g_j
  },
  \qquad j=0,\dots,J.
\]
Hence solving the linear Volterra equation using this scheme amounts to a
lower-triangular forward substitution.

To make this scheme explicit, it thus suffices to compute the weights
\(\omega^L\) and \(\omega^R\) for the fractional-logarithmic kernels appearing
in the expansion. After the change of variables \(u=t_j-s\), these integrals
reduce to primitives of
\[
  u^{H_0-\frac12}\log^n(u),
\]
which are obtained recursively from
\[
  \int u^{H_0-\frac12}\log^n(u)\dd{u}
  =
  \frac{u^{H_0+\frac12}}{H_0+\frac12}\log^n(u)
  -
  \frac{n}{H_0+\frac12}
  \int u^{H_0-\frac12}\log^{n-1}(u)\dd{u}.
\]
Hence all quadrature weights can be precomputed once for the chosen grid and
then reused throughout the expansion recursion.

\begin{Rmk}
Product integration goes back to \cite{Young1954application}, and the product
trapezoidal rule employed here is classical for second-kind Volterra equations
with weakly singular kernels; see
\cite{Linz1985,deHoogWeiss1974,BrunnerVanDerHouwen1986}. For
the pure fractional kernel ($n=0$) with smooth data, the scheme is second-order
accurate \cite{deHoogWeiss1974}; for the non-smooth solutions typical of this
setting, the attainable order on uniform meshes is $2-\alpha$ away from the
origin \cite{Dixon1985}, with the global rate reduced to $1+\alpha$,
$\alpha=H_0+\tfrac12$, see also \cite{DiethelmFordFreed2004,Garrappa2015}.
Our setting falls outside the hypotheses of these results in two respects:
the kernels carry mixed fractional-logarithmic singularities
$t^{\alpha-1}\log^n t$, and the coefficient $g(\cdot;z)$ and inhomogeneity
$r$ inherit from the expansion terms a power-logarithmic behavior of the form
$t^{j\alpha+m}\log^k t$ near the origin, rather than the smoothness assumed in
the classical theory.
Nevertheless we expect the rates of the pure fractional
case to persist up to logarithmic factors.
\end{Rmk}

\subsection{An implicit product-integration scheme for the fractional Riccati equation}
\label{sec:implicit_riccati_pi}

In this section we describe a fully implicit product-integration discretization of the fractional Riccati equation \eqref{eq:riccati}. While implicit product-integration rules for weakly singular Volterra equations are classical \cite{Linz1985,JumarhonMcKee1996,Garrappa2018}, the nonlinear step is usually closed by Newton or fixed-point iteration. Here, the quadratic Riccati driver makes each update an explicitly solvable scalar quadratic equation.

Closely related are the ``iVi'' simulation schemes, where the Inverse Gaussian increment laws can be recovered from right-endpoint discretizations of the corresponding conditional Riccati and Riccati--Volterra equations, with the quadratic root having non-positive real part selected
\cite[Remark~1.6]{AbiJaber2024Simple}
\cite[Remark~1.7]{AbiJaberAttal2025}.
In \cite{XieGao2025}, explicitly solvable quadratic updates are likewise obtained from an implicit finite-difference discretization of a real-valued fractional Riccati equation, with the positive branch selected.
We select the root closest to the linearized implicit update. This is reminiscent of predictor--corrector continuation for algebraic Riccati equations, where a linearized solution provides a predictor for the nearby nonlinear branch \cite{DieciFriedman2001}, as well as of complex-branch tracking in Heston-type formulas \cite{Albrecher,KL06,KL10}.

To the best of our knowledge, this combination of fully implicit product integration, exact quadratic updates, and linearized branch selection has not been studied so for the deterministic fractional Riccati equation. In our tests, it is substantially more stable than the predictor--corrector Adams update in the hyper-rough regime $H<0$, especially for large Fourier frequencies.
A rigorous convergence analysis is left for future work; since the implicit
equation is solved exactly, the rates should follow from the classical theory
of implicit product-integration and collocation methods for weakly singular
Volterra equations
\cite{JumarhonMcKee1996,Dixon1985,BrunnerPedasVainikko1999,Garrappa2015},
adapted to the regularity of \(\psi^H\) at the origin and with constants
depending on \(z\); see also \cite{DiethelmFordFreed2004,ShiLyuMa2022}.

Fix a time grid \(0=t_0<t_1<\cdots<t_J=T\) and let \(\Omega\) be the
product-integration matrix defined in the previous section on this grid, now
for the fractional kernel \(K(t)=t^{H-\frac12}/\Gamma(H+1/2)\)
specifically; recall that \(H>-1/2\), so that \(K\) is integrable and
\(\Omega\) is well defined, with diagonal entries
\(\Omega_{j,j}=(t_j-t_{j-1})^{H+\frac12}/\Gamma(H+5/2)\). Throughout,
\(z\in\mathbb{C}^{\ast}\) is fixed; rhe computation is easily vectorized over \(z\), as
required for the evaluation of \(\psi^H(T,\cdot)\) along a Fourier pricing
contour.
Further assume that the mesh is fine enough that
\[
  \Omega_{j,j}\,\lvert\rho\nu z-\lambda\rvert<1,
  \qquad j=1,\dots,J.
\]
Recalling the definition of the Riccati driver \(F\) in
\eqref{def:F}, the fully implicit product-integration discretization is
obtained by piecewise-linear interpolation of the smooth factor
\(s\mapsto F(z,\psi^H(s,z))\) of the integrand, i.e.,
\[
  \psi_j
  =
  \sum_{\ell=0}^j
  \Omega_{j,\ell}F(z,\psi_\ell),
  \qquad
  \psi_j\approx \psi^H(t_j,z),
  \qquad
  \psi_0=0.
\]
Equivalently, separating the known history from the diagonal contribution
gives
\[
  \psi_j
  =
  h_j
  +
  \Omega_{j,j}F(z,\psi_j),
  \qquad
  h_j
  :=
  \sum_{\ell=0}^{j-1}
  \Omega_{j,\ell}F(z,\psi_\ell).
\]
Since \(F\) is quadratic in its second argument, the implicit update at
\(t_j\) is a scalar quadratic equation for \(\psi_j\), whose roots we write
as
\[
  \psi_j^{\pm}
  =
  \frac{
    2\left(h_j+\frac12\Omega_{j,j}\,z(z-1)\right)
  }{
    \big(1-\Omega_{j,j}(\rho\nu z-\lambda)\big)
    \pm
    \sqrt{
      \big(1-\Omega_{j,j}(\rho\nu z-\lambda)\big)^2
      -2\,\Omega_{j,j}\nu^2\left(h_j+\frac12\Omega_{j,j}\,z(z-1)\right)
    }
  },
\]
where the branch is selected such that \(\psi_j\) is closest to
the linearized implicit update
\[
  \psi_j^{\mathrm{lin}}
  =
  \frac{
    h_j+\frac12\Omega_{j,j}\,z(z-1)
  }{
    1-\Omega_{j,j}(\rho\nu z-\lambda)
  },
\]
obtained by dropping the quadratic term of \(F\). The
selection rule picks the regular branch of a singularly perturbed quadratic:
as the mesh is refined, \(\Omega_{j,j}\to0\) and one root has the same limit as 
\(\psi_j^{\mathrm{lin}}\), while the other
diverges like \(1/\Omega_{j,j}\), and only the former is a consistent
approximation of \(\psi^H(t_j,z)\)---the unique solution of the implicit
update that remains bounded under refinement, and the limit of a fixed-point
iteration started at \(h_j\). Since the correct branch may switch as \(z\)
and \(t_j\) vary in the complex case, the comparison with
\(\psi_j^{\mathrm{lin}}\) provides an adaptive identification.

\subsection{Root-Padé approximation of the fractional Riccati solution}
\label{app:root-pade}

Rational approximations of the fractional Riccati solution go back to
Gatheral and Radoi\v{c}i\'c
\cite{gatheral_radoicic_2019,gatheral_radoicic_2024}, who match low-order
rational functions to both the short-time expansion and the long-time
asymptotics around the stable root, in the spirit of global Pad\'e
approximations of Mittag-Leffler functions
\cite{AtkinsonOsseiran2011,ZengChen2015,SarumiFuratiKhaliq2020}. The
construction used here is somewhat less sophisticated: after a M\"obius
transformation of the Riccati variable, we match a high-order rational
function to the short-time expansion and impose only the leading long-time
coefficient, using the natural scaled variable
$x=\Delta(z)t^{\alpha}$. In the classical case, the transformed variable is
exactly $e^{-x}$; replacing it by $E_{\alpha}(-x)$---a substitution known from
the fractional logistic equation to be inexact for $\alpha<1$
\cite{West2015,AreaLosadaNieto2016}---nevertheless reproduces the leading
long-time ansatz of \cite{gatheral_radoicic_2024}. The transformation enforces
the initial value and stationary limit, while the rational degrees enforce
the decay order. All but one matching condition come from the short-time
expansion, available to arbitrary order from a cheap recursion, and the
remaining one fixes the leading approach to the stationary root. This avoids
higher-order long-time coefficients, possibly at some cost in accuracy at
large but finite maturities; related allocations favouring short-time
conditions are studied in \cite{JengKilicman2020math}.
Within the parameter range considered in our numerical experiments, in particular for $H\in[-0.4,0.4]$, the resulting root-Pad\'e approximant of degree $m=7$ improved over the $[4/4]$ Pad\'e approximation of \cite{gatheral_radoicic_2019}, particularly in the out-of-the-money wings. A detailed error analysis is left for future work.

Throughout, we use the notation of \eqref{eq:riccati}. Fix the expansion
anchor $H_0$ and set $\alpha=H_0+1/2$. Recall from
\cite[Sections~2--3]{gatheral_radoicic_2024} that the short-time behaviour of
$\psi^{H_0}(t,z)$ is naturally expanded in powers of $y=t^{\alpha}$,
\[
  \psi^{H_0}(t,z)=\sum_{n\ge1}b_n(z)\,y^{n},
\]
where
\[
  b_1(z)=\frac{z(z-1)}{2\,\Gamma(\alpha+1)},
  \qquad
  b_{n+1}(z)
  =
  \frac{\Gamma(n\alpha+1)}{\Gamma((n+1)\alpha+1)}
  \left(
    (\rho\nu z-\lambda)\,b_n(z)
    +\frac{\nu^2}{2}\sum_{k=1}^{n-1}b_k(z)\,b_{n-k}(z)
  \right),
\]
for all $n\ge1$, which determines the short-time expansion to arbitrary
order. The long-time behaviour is governed by the stationary roots of
$F(z,\cdot)$,
\[
  r_\pm(z)=\frac{-(\rho\nu z-\lambda)\pm\Delta(z)}{\nu^2},
  \qquad
  \Delta(z)=\sqrt{(\rho\nu z-\lambda)^2-\nu^2\,z(z-1)},
\]
where the branch $\Re\Delta(z)\ge0$ makes $r_-(z)$ the stable
root, and takes the form
\[
  \psi^{H_0}(t,z)\sim r_-(z)+\frac{g_1(z)}{y}+\frac{g_2(z)}{y^2}+\cdots,
  \qquad y\to\infty.
\]

At this point, we depart from the direct rational approximation of
$\psi^{H_0}$ in \cite{gatheral_radoicic_2019,gatheral_radoicic_2024}.
We first apply the root-adapted fractional linear transformation that
linearizes the classical ($\alpha=1$) Riccati flow: with
$A(z):=r_-(z)\,r_+(z)=z(z-1)/\nu^2$, set
\[
  E(t,z)
  =
  \frac{A(z)-r_+(z)\,\psi^{H_0}(t,z)}{A(z)-r_-(z)\,\psi^{H_0}(t,z)},
  \qquad\text{equivalently}\qquad
  \psi^{H_0}(t,z)
  =
  A(z)\,\frac{1-E(t,z)}{r_+(z)-r_-(z)\,E(t,z)}.
\]
Since $\psi^{H_0}(0,z)=0$, we have $E(0,z)=1$, and convergence of
$\psi^{H_0}$ to the stable root is equivalent to $E(t,z)\to0$: thus $E$
describes the relaxation from the initial value to the stationary Riccati
root, with $E(t,z)=\exp({-\Delta(z)t})$ exactly in the classical case. Formal
power-series division yields the expansion
$E(t,z)=1+\sum_{n\ge1}e_n(z)\,x^{n}$ in the scaled variable
\[
  x=\Delta(z)\,t^{\alpha}.
\]
For
$0<\alpha<1$, the leading long-time term implies
$E(t,z)\sim C_\alpha(z)/x$, where
$$C_\alpha(z)=\frac{r_+(z)}{(r_+(z)-r_-(z))\Gamma(1-\alpha)}.$$
We then approximate $E$ by a rational function of type $[m/(m{+}1)]$,
\[
  E(t,z)\approx\frac{P_m(x,z)}{Q_{m+1}(x,z)},
  \qquad
  Q_{m+1}(x,z)\,E(x,z)-P_m(x,z)=\mathcal{O}\big(x^{2m+1}\big),
  \qquad
  \frac{p_m(z)}{q_{m+1}(z)}=C_\alpha(z),
\]
where $p_m(z)$ and $q_{m+1}(z)$ denote the leading coefficients of
$P_m(\cdot,z)$ and $Q_{m+1}(\cdot,z)$, respectively, and
$P_m(0,z)=Q_{m+1}(0,z)=1$. Thus the highest short-time matching condition is
replaced by the exact leading tail coefficient. As the denominator degree
exceeds the numerator degree by one, $P_m/Q_{m+1}\to0$ as $x\to\infty$,
matching the $\mathcal{O}(y^{-1})$ decay of $E$. The root-Pad\'e approximation
is therefore
\[
  \psi^{H_0}(t,z)
  \approx
  A(z)\,
  \frac{1-P_m(x,z)/Q_{m+1}(x,z)}
       {r_+(z)-r_-(z)\,P_m(x,z)/Q_{m+1}(x,z)},
  \qquad
  x=\Delta(z)\,t^{\alpha}.
\]

\end{document}